\documentclass[reqno,11pt]{amsart}
\usepackage{DocumentPreamble}

\AddTitle{Spectral graph clustering with inhomogeneous latent geometry}

\AddAuthor{Konstantin Avrachenkov\textsuperscript{$\ast$}}
\AddAuthor{Lucas S. Sibemberg\textsuperscript{$\dagger$}}
\AddAuthor{Alexander Van Werde\textsuperscript{$\diamond$}}
 
\address{$\ast$ Inria Sophia Antipolis, France
  (\href{mailto:k.avrachenkov@inria.fr}{k.avrachenkov@inria.fr})}
\address{$\dagger$ Federal University of Rio Grande do Sul, Brazil (\href{mailto:lucas.siviero@ufrgs.br}{lucas.siviero@ufrgs.br})}
\address{$\diamond$ University of Münster, Germany
  (\href{mailto:a.van.werde@uni-muenster.de}{a.van.werde@uni-muenster.de})
  } 

\AddMSC{05C80}
\AddMSC{60B20}
\AddMSC{62H30}
\AddMSC{05C82}

\AddKeyword{Block model} 
\AddKeyword{almost exact recovery} 
\AddKeyword{density-based spectral clustering} 
\AddKeyword{geometric random graph}  

\begin{document}
\begin{abstract}    

We study spectral clustering in the presence of a confounding latent geometry. 
The leading eigenvectors may then be dominated by the latent geometry rather than by the communities. 
Nevertheless, we show in a \emph{block latent-space model} that communities can be recovered from eigenvectors deeper in the spectrum. 
We analyze the spectral properties of the adjacency matrix through a limiting integral operator and use its structure to develop DBSPEC, a density-based spectral clustering algorithm that requires only approximate localization of the informative eigenvalue and is robust to poor eigenvalue separation. 
Crucially, this approach handles general latent geometries, overcoming restrictions to homogeneous toroidal models in prior works. 
Our theoretical predictions for the location of the informative eigenvalue notably align with observations in real-world experiments.
\end{abstract}
\maketitle

\section{Introduction}

Many real-world networks are shaped by both geometric and community-based factors.
For example, consider a collaboration graph where an edge connects two researchers who co-author a paper.
Shared research interests increase the chance of collaboration, while geographic distance reduces it due to logistical challenges. 
Such effects are indeed apparent in empirical data:  researchers from the same domain but different institutions often collaborate despite 
geographic distance, whereas collaborations across domains are more frequent 
within the same institution \cite[Table 1]{galhotra2023community}.
Thus, collaboration networks reflect both community structure (research domains) and geometric structure (spatial proximity).

To uncover such underlying structures from empirical data, clustering algorithms that group together similar nodes serve as a key tool.
Various clustering algorithms have been developed depending on the available data and on the structure that one is interested in. 
If one is predominantly interested in the spatial relationships, and one has access to the coordinates, then distance-based algorithms such as $k$-means or \texttt{DBSCAN} may be used \cite{ester1996density,lloyd1982least}.
If instead one is interested in the community structure, as will be our focus, then algorithms that leverage network structure are more suitable.
Among these, spectral algorithms are especially popular due to their strong theoretical foundations and ease of implementation.

Most fundamental understanding of network clustering algorithms comes from theoretical models that include either geometric or community structure, but not both.
For instance, the \emph{stochastic block model} captures purely community-based structure, with edges formed according to group memberships, while ignoring other factors such as geometry \cite{holland1983stochastic}. 
Spectral clustering methods that rely on the eigenvectors associated with the leading few eigenvalues of a Laplacian or adjacency matrix can be shown to be consistent in that model because the observed matrices can be viewed as a low-rank signal plus noise \cite{lei2015consistency}. 
By contrast, \emph{latent position models} provide a purely geometric framework for clustering, in which each node is assigned a position in Euclidean space and edges are formed according to distance \cite{handcock2007model}. Spectral methods are also consistent in this setting, but for a different reason: the graph Laplacian approximates the continuous 
Laplace--Beltrami operator, whose eigenfunctions capture the latent geometry \cite{von2008consistency}.

In practice, however, geometry and community effects often 
coexist. 
Spatial structure can confound algorithms designed for community detection. 
For example, we show in Section \ref{sec: RealWorld} that direct application of standard spectral clustering to a collaboration network may fail to recover the underlying research domains when the effect of distance is strong.
The leading eigenvectors are then dominated by the spatial structure which obscures the desired non-geometric clusters.

To address such issues, one must account for geometric effects.  
To this end, the two-cluster \emph{Soft Geometric Block Model} (SGBM) assigns each node a uniform random position $X_i$ on a $d$-dimensional torus as well as a cluster label in $\{1,2 \}$ \cite{avrachenkov2021higher}. 
The edge probabilities in this random graph model depend on the distance between nodes through the function $F_{\textnormal{in}}(\Vert X_i-X_j\Vert)$ if the nodes belong to the same cluster, and through a different function $F_{\textnormal{out}}(\Vert X_i-X_j\Vert)$ if they belong to different clusters.
This model generalizes the \emph{hard geometric block model} whose study was initiated by \cite{galhotra2018geometric,galhotra2023community}.

It was found in \cite{avrachenkov2021higher} that sign-based spectral clustering can still be made to work in the SGBM, but only if one considers an eigenvalue deeper in the spectrum. 
To locate the appropriate eigenvalue, \cite{avrachenkov2021higher} analyzed the limiting spectrum through a delicate combinatorial approach based on the tracial moments of the adjacency matrix. 
The limiting spectrum there appears somewhat miraculously via a Fourier transform of $F_{\textnormal{in}}$ and $F_{\textnormal{out}}$. 
This analysis strongly depends on the nodes being distributed uniformly in the $d$-dimensional torus. 
Furthermore, the exact appropriate eigenvalue position must be known and the algorithm is not robust when the eigenvalue is not well-separated from other eigenvalues.

Our goal in the present paper is to understand spectral clustering techniques in the presence of inhomogeneous latent geometry.
We study a \emph{block latent-space model} where nodes are independently assigned positions $X_i$ on some arbitrary compact set following a law that does not need to be uniform.   
We let the connection probabilities be proportional to $K(X_i,X_j)$, where $K$ is a general symmetric
kernel satisfying a mild $L^2$-continuity condition. The proportionality constant of the kernel depends
on whether the nodes are in the same cluster or not.
Thus, the latent geometry need not be a torus, and the kernel may depend on factors beyond distance; however, its dependence on the cluster is assumed to enter only through the proportionality constants.
This model hence generalizes the case of the SGBM with $F_{\textnormal{in}} = a F$ and $F_{\textnormal{out}} = b F$ for constants $a,b \in [0,1]$ and a single fixed function $F$. 

We analyze the limiting eigenvalues and eigenvectors in Propositions \ref{prop: EigenvaluesMatch} and \ref{prop: EigenvectorsMatch}. 
Instead of combinatorial analysis of tracial moments, we rely on results from \cite{koltchinskii1998asymptotics} to match the spectral properties with a limiting integral operator.
This approach not only handles inhomogeneous latent geometry, but also provides additional insight through the appearance of the spectrum as that of a limiting operator, and it lends itself to generalizations; see e.g., Remark \ref{rem: MutlipleClusters}.

\pagebreak[4]
The inhomogeneous latent geometry implies that the limiting eigenvectors do not need to be constant on the clusters, different from what occurs in the SGBM in \cite{avrachenkov2021higher}.  
Rather, the relevant limits are continuous functions when restricted to a single cluster but may be discontinuous across clusters. The proposed algorithm \texttt{DBSPEC} (Algorithm~\ref{alg:high-density}) leverages this by using eigenvector coordinates to construct low-dimensional embeddings and then applying \texttt{DBSCAN} to separate them into two connected parts.
This algorithm has the additional advantage of being more robust to eigenvalue selection than the method of \cite{avrachenkov2021higher}, since the continuity of the latent positions is preserved even when a few irrelevant eigenvectors are included in the embedding. Moreover, the use of low-dimensional eigenspaces rather than individual eigenvectors makes the method naturally robust to eigenvalue multiplicities.

Theorem \ref{thm: AlgGuarantee} gives an algorithmic guarantee that almost exact recovery is possible given only approximate knowledge of the ideal eigenvalue positions. 
This is established in a regime with superlogarithmic average degree, which is more lenient than the assumption of linear average degree in \cite{avrachenkov2021higher}. 
This assumption is expected to be close to optimal for this specific algorithm. 
Additional regularization steps in the spirit of \cite{le2017concentration} could likely enable algorithms that extend almost-exact recovery to sparser regimes. 
We do not pursue such additional steps here as they are not directly related to our primary goal of understanding the effect of inhomogeneous latent geometry on spectral clustering. 

The theoretical insights gained in our block latent-space model are leveraged in Section~\ref{sec: RealWorld} to extract community structure despite the confounding background geometric effects in some real-world networks. 
There, we show that standard spectral methods may fail to identify the communities of interest when the leading spectral information is dominated by geometry rather than cluster structure. In such cases, better performance can be achieved by exploiting a non-leading eigenvalue.
Most notably, our theoretical predictions of the location of the ideal eigenvalue match the observed ideal eigenvalues in these empirical experiments.

Let us stress that our algorithms do not assume that the latent positions are known; only the network structure is observed.
An estimator for the ideal eigenvalue position does require some additional information, but even this is only about cluster-level parameters; see \eqref{eq:GrumpySax} and the preceding discussion. 
This has significant relevance in the empirical applications as we indeed lack prior knowledge of the specific latent geometric distributions in those examples.

\subsection*{Additional literature}
The two-cluster SGBM and clustering algorithms from \cite{avrachenkov2021higher} were recently extended in \cite{allem2025multi} to multiple clusters by using $k$-means instead of sign-based clustering. 
The latter comes with similar restrictions to \cite{avrachenkov2021higher}: the geometry is toroidal and the algorithm is non-robust when eigenvalue separation is small.

Clustering when latent positions are known has been considered in a stream of works, starting from \cite{sankararaman2018community,abbe2021community}.
In \cite{sankararaman2018community} information-theoretic conditions for weak recovery (detection) were studied, and \cite{abbe2021community} studied both weak and exact recovery. 
More recently,  \cite{gaudio2024exacta,gaudio2024exactb,avrachenkov2024community,gaudio2025sharp,gaudio2025exact} refined the information-theoretic conditions for exact recovery and proposed efficient algorithms based on local, geometry-aware propagation and refinement procedures. 
The papers \cite{gaudio2024exacta,gaudio2024exactb,gaudio2025sharp} focus on hard random geometric graph models, whereas \cite{avrachenkov2024community,gaudio2025exact} study soft geometric models with distance-dependent connection probabilities.
In all the above works, nodes are assumed uniform on a torus, an assumption that is used both to obtain tractable information-theoretic thresholds and to analyze the algorithms. 
It would be interesting to also pursue sharp information-theoretic thresholds in our setting with unobserved and inhomogeneous latent positions.

There has also been a significant recent work on the spectra of geometric graph models without cluster structure \cite{cao2025spectra,ding2025edge,adhikari2022spectrum,bordenave2008eigenvalues,bubeck2016testing,dubova2023universality}. 
Topics of interest include the effect of high-dimensional latent geometry or sparsity. 
It would also be interesting to understand the effect of such features on clustering problems.  

\subsection*{Outline}
Section \ref{sec: Model} introduces the model. 
We state our results regarding the spectral properties and clustering algorithms in Sections \ref{sec: LimitSpec} and \ref{sec: SpectralClusteringAlgorithm}, respectively. 
Proofs are given in Sections \ref{sec: ProofOperator} and \ref{sec: ProofAlgGuarantee}. 
We conclude with experiments on real-world data in Section \ref{sec: RealWorld}. 
Further details, including an outline of the extension to more than two clusters and additional experiments on synthetic and real-world data are provided in the appendices.

    \section{Model definition and assumptions}\label{sec: Model} 
    We consider a random graph model whose edge probabilities depend on latent variables and cluster labels.
    
    Fix a compact subset $\cX \subseteq \bbR^d$ for some $d\geq 1$ and assume that the nodes are independently assigned latent positions $X_1,\ldots,X_n$ in $\cX$ following a fixed probability distribution $\mu$. 
    Further, independently assign random labels $\sigma_1,\ldots,\sigma_n\sim  \operatorname{Unif}\{1,2 \}$.
    Generalizations with more clusters and non-uniform assignment are considered in Remark \ref{rem: MutlipleClusters} and the appendices. 

    To specify the dependence of the connection probabilities on latent positions, fix a measurable function $K:\cX\times \cX \to [0,1]$ with the symmetry constraint $K(x, y) = K(y,x)$.
    Further, for the dependence on the cluster labels, fix $a,b\in [0,1]$ and define a symmetric $2\times 2$ matrix 
    \begin{align} 
        \bP \de \begin{pmatrix}
            a&b\\ 
            b & a
        \end{pmatrix}.\label{eq: Def_P}
    \end{align}
    Finally, to tune the graph's sparsity, let $\rho_n\in[0,1]$ be an $n$-dependent parameter.
    
    A random graph $G = (V,E)$ on $n$ nodes is then said to come from the \emph{block latent-space model} if its edges are present conditionally independently given the latent positions $X_v$ and cluster assignments $\sigma_v$, with conditional probability specified by 
    \begin{align} 
        \bbP_{X,\sigma}\bigl((v,w)\in E \bigr) = 
            \rho_n K(X_v,X_w) \bP(\sigma_v,\sigma_w).\label{eq: Def_PXsigma}
    \end{align}
    Equivalently, the conditional distribution of the adjacency matrix $\bA\in \{0,1 \}^{n\times n}$ is given by  
    \begin{align} 
        \bbP_{X,\sigma}(\bA) = \prod_{1 \leq v<w \leq n}  \bigl(
            \rho_n K(X_v,X_w) \bP(\sigma_v,\sigma_w) \bigr)^{\bA_{v,w}} \bigl(1 - 
            \rho_n K(X_v,X_w) \bP(\sigma_v,\sigma_w)\bigr)^{1-\bA_{v,w}}. 
    \end{align}
    \begin{example}
        If $\cX = \{0 \}$ is a single point, then the connection probabilities only depend on the clusters assignments and one recovers the \emph{stochastic block model} \cite{holland1983stochastic}. 
        Moreover, if $\cX \subseteq \bbR$ and $K(x,y) = xy$, then one recovers the \emph{degree-corrected stochastic block model} \cite{qin2013regularized,karrer2011stochastic}.
    \end{example}
    \begin{example}\label{example: Geometric}
        Let $K(x,y) = \bb1\{\Vert x-y \Vert \leq R\}$ for some $R>0$. 
        Then, one recovers the \emph{geometric random graph model} \cite{penrose2003random} as the special case where $a=b$ in \eqref{eq: Def_P}. 
        Moreover, by allowing $a$ and $b$ to be different one gets a notion of cluster structure, giving an instance of the \emph{soft geometric block model} studied in \cite{avrachenkov2021higher} in the special case where $\cX$ is a torus.
        The restriction to tori will not be needed in the present work. 
        A visualization of the model in the case, when $\cX = [0,1]\times [0,1]$ is the square, is depicted in Figure \ref{fig:01}. 
        Note that this latent space introduces a boundary effect which makes the geometry of the graph non-homogeneous and was not present in \cite{avrachenkov2021higher}: nodes near the edge have fewer neighbors than those in the middle of the domain.
    \end{example}

    \begin{figure}[h!]
      \centering
      \includegraphics[width=0.68\linewidth]{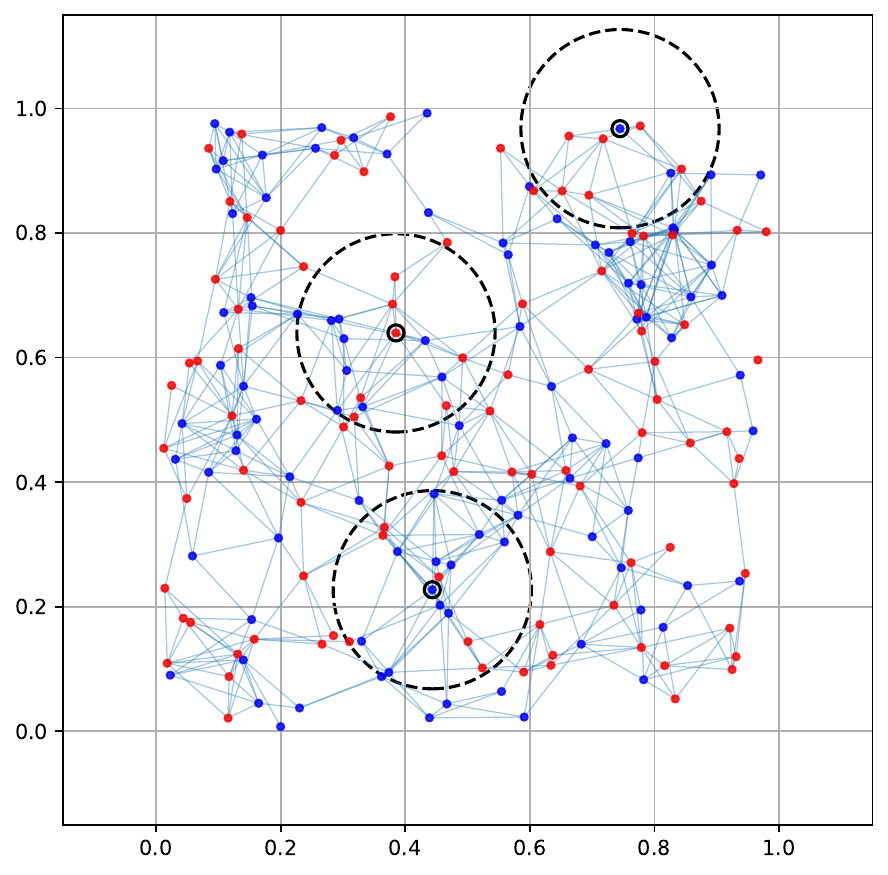}
      \caption{A random soft geometric graph with $200$ points uniformly drawn in the square $\cX$ as in Example~\ref{example: Geometric} with parameters $R=1/2\pi$, $a=0.75$ and $b=0.2$. 
      Community membership is indicated by node colors. 
      }
      \label{fig:01}
    \end{figure}
    \begin{nonexample}
        We stress that the present setup is intended to model a setting where clusters and geometry are two different features. 
        This is reflected in our assumption that the latent probability distribution $\mu$ and the kernel $K$ are the same for both clusters.
        In particular, geometry-based clusters as in a Gaussian mixture block model graph \cite{li2023spectral} are a different notion and are not admitted in our model. 
        Hard geometric block models with cluster-dependent connection kernels as in \cite{galhotra2018geometric,galhotra2023community} are not admitted in our model either. 
    \end{nonexample}
        \subsection{Assumptions}
    So far, the cluster assignments $\sigma_v$ could just as well have been incorporated in the latent positions by replacing the latent space $\cX$ with two disjoint copies of itself and modifying the kernel $K$.
    The following assumption rules this out to make the clustering problem well-defined. 
    \begin{assumption}\label{ass: ConnectedLatent}
        The compact set $\cX \subseteq \bbR^d$ is connected. 
        Moreover, the support of the measure $\mu$ completely covers this set: $\operatorname{supp}(\mu) = \cX$.   
    \end{assumption}
    An assumption on $K$ is required to ensure that it respects the underlying geometry of the latent space. 
    To this end, we will rely on a continuity assumption. 
    It is however too restrictive to impose that $K$ is continuous as a bivariate function as this would rule  out important examples like block geometric random graphs; recall Example \ref{example: Geometric}.
    The following  is more flexible:
    \begin{assumption}\label{ass: ContinuityKernel}
        The function $x \mapsto K(x, \cdot)$ is continuous in the $L^2$-norm for  $\mu$. 
        That is, $\lim_{\varepsilon \to 0} \int_\cX ( K(x,y) - K(x+\varepsilon,y) )^2 \intd \mu(y) =0$ for every fixed $x\in \cX$.  
    \end{assumption}
  
    Assumptions \ref{ass: ConnectedLatent} and \ref{ass: ContinuityKernel} together encode the conceptual difference between clusters and latent variables that we assume in this paper: while the cluster assignments are discrete in nature, the latent space $\cX$ is assumed to have a continuous character. 
    Our focus will be to understand how spectral algorithms for recovering the cluster labels are affected by the confounding latent variables.      
     
    Of course, it may not be possible to recover the cluster label $\sigma_v$ based on an observation of the graph $G$ if the associated vertex has no connections that we can use.
    Further, the two clusters are statistically identical if $a=b$. 
    This motivates the following assumptions:
    \begin{assumption}\label{ass: Nontrivial}
        It holds that $a+b>0$ and $a-b \neq 0$. 
        Further, assume that for every $x\in \cX$ it holds that  
        $
        \int_{\cX} K(x,y)\, \intd\mu(y) >0 
        $.  
    \end{assumption}
    All results presented from here on suppose that Assumptions \ref{ass: ConnectedLatent} to \ref{ass: Nontrivial} are in effect. 
    \section{Limiting spectral properties}\label{sec: LimitSpec}
    Associated with the kernel $K$, we have a linear operator 
    $
        \mathbb{K}:L^2(\cX, \mu) \to L^2(\cX,\mu)
    $ 
    given by 
    \begin{align} 
        (\mathbb{K}f)(x) \de \int_{\cX} f(y) K(x,y)\, \intd\mu(y).    
    \end{align} 
    It follows from $K$ being symmetric and taking values in $[0,1]$ that the operator $\mathbb{K}$ is self-adjoint and compact. 
    In particular, the spectral theorem applies \cite[Ch.~II, \S5, Theorem~5.1]{conway2007course}: there exists an orthonormal system of eigenfunctions $\varphi_1,\varphi_2,\ldots \in L^2(\cX,\mu)$ with eigenvalues $\lvert\kappa_1 \rvert \geq \lvert \kappa_2 \rvert \geq \ldots$ whose only accumulation point is at zero such that for every $f\in L^2(\cX,\mu)$,
    \begin{align} 
        \mathbb{K}f = \sum_{i=1}^\infty \kappa_i\langle f, \varphi_k \rangle \varphi_k. \label{eq:UneasyCamel} 
    \end{align}
    We next show that the limiting eigenvectors and eigenvalues of the adjacency matrix can be computed as a transformation of those of $\mathbb{K}$. 
    \begin{proposition}\label{prop: EigenvaluesMatch}
        Let $\hat{\lambda}_1,\ldots,\hat{\lambda}_n$ be the eigenvalues of the rescaled adjacency matrix $\bA/(n\rho_n)$ and assume that $\rho_n = \omega(\ln(n)/n)$.
        Then, there exists a measure $\nu$ such that for any fixed Borel set $B$ with $\nu(\partial B) = 0$ and $0\not\in \overline{B}$, we have 
        \begin{align} 
            \#\{i \leq n:  \hat{\lambda}_i \in B \}
 \to   \nu(B)
        \end{align}
        in probability.
        Moreover, this limiting measure $\nu$ is explicitly given by 
        \begin{align} 
            \nu = \sum_{i=1}^\infty \delta_{\frac{(a+b)}{2}\kappa_i}  + \delta_{\frac{(a-b)}{2}\kappa_i}.\label{eq:MadImp}
        \end{align}
    \end{proposition}
    Proposition \ref{prop: EigenvaluesMatch} can be viewed as a generalization of \cite[Theorem 1]{avrachenkov2021higher} and \cite[Theorem 2.1]{allem2025multi} which concern cases where $\cX$ is a torus. 
    Our proof approach, however, is different.
    While the analysis in \cite{allem2025multi,avrachenkov2021higher}  used delicate combinatorial arguments for the tracial moments of the adjacency matrix where the shape of the limiting spectrum appears somewhat miraculously, here we adopt an operator-theoretic perspective. 
    Our proof in Section \ref{sec: ProofOperator}, shows that the adjacency matrix has spectral properties similar to an explicit limiting operator $\bbA$ with a spectrum given by \eqref{eq:MadImp}. 

    Specifically, this operator $\bbA$ can be viewed as a kernel operator on the discrete union of two copies of $\cX$. 
    Suppose that we are given a pair of functions $f_1,f_2:\cX \to  \bbR$ in $L^2(\cX,\mu)$, one for each copy of $\cX$. 
    Then, we define the action of $\bbA$ as follows:  
    \begin{align} 
            \mathbb{A}(f_1, f_2) \de (g_1, g_2) \ \textnormal{ with }\ g_i(x) \de \frac{1}{2}\sum_{j=1}^2\bP(i,j)\int_{\cX}  K(x,y) f_j(y) \, \intd \mu(y).  \label{eq: Def_bbA} 
    \end{align}
    Viewing the two copies of $\cX$ as corresponding to the different clusters and recalling the definition of the connection probabilities in the model from \eqref{eq: Def_PXsigma}, one can interpret \eqref{eq: Def_bbA} as a mean-field approximation to the action of the rescaled adjacency matrix $\bA/(n\rho_n)$.\footnote{This mean-field approximation is similar in spirit to the now-classical result that a random graph sampled from a graphon converges to it, as do the eigenvalues \cite[\S 11.6]{lovasz2012large}. 
    Actually, the specific results from integral operator theory \cite{koltchinskii1998asymptotics,koltchinskii2000random} that we build our proofs on even predate the graphon-theoretic terminology.}

    Not only does this explain \emph{why} the limiting spectrum has the form specified in \eqref{eq:MadImp}, one can also use this approach to deduce the structure of the eigenvectors.
    A direct calculation with \eqref{eq: Def_bbA} verifies that the eigenfunctions of $\bbA$ are $(\varphi_i, \varphi_i)$ and $(\varphi_i, -\varphi_i)$ with eigenvalues $(a+b) \kappa_i/2$ and $(a-b)\kappa_i/2$, respectively. 
    The eigenvectors of the adjacency matrix mimic this. 
    Conditional on the latent positions and the cluster assignments, let us define vectors $\psi_i^{+},\psi_i^-\in \bbR^n$ by  
    \begin{align} 
        (\psi_i^+)_v \propto  \varphi_i(X_{v})\quad \textnormal{ and } \quad 
        (\psi_i^-)_v \propto 
        \begin{cases}
            \varphi_i(X_v) & \textnormal{ if } \sigma_v = 1\\ 
            - \varphi_i(X_v)&\textnormal{ if }\sigma_v = 2
        \end{cases}\label{eq: Def_psi}
    \end{align}
    where the proportionality constant is chosen so that these are unit vectors.\footnote{To make sense of pointwise evaluation of $\varphi_i$ despite the equivalence relation implicit in $L^2(\cX, \mu)$, let us note that it follows from Assumption~\ref{ass: ContinuityKernel} that these eigenfunctions are continuous functions; see Lemma \ref{lem: Continuity}.} 
    Then, we have the following approximation result:  
    \begin{proposition}\label{prop: EigenvectorsMatch}
        Fix a Borel set $B$ with $\nu(B)>0$. 
        Further, assume that $\nu(\partial B) = 0$ and $0 \not\in \overline{B}$. 
        Let $\Psi_B$ be the $n\times \nu(B)$ matrix with columns given the $\psi_i^\pm$ with $(a\pm b)\kappa_i/2 \in B$, and let $\hat{\Psi}_B$ be a matrix with columns given by those eigenvectors of $\bA/(n\rho_n )$ with eigenvalues in $B$.

        Assume that $\rho_n = \omega(\ln(n)/n)$.
        Then, it holds asymptotically almost surely that $\hat{\Psi}_B$ is also $n\times \nu(B)$. 
        Moreover, it holds in probability that 
        \begin{align} 
            \min_{\hat{\bO} \in \mathbb{O}(\nu(B))}\Vert \hat{\Psi}_B \hat{\bO} - \Psi_B \Vert_{F} \to 0\label{eq:Procruste} 
        \end{align}
        where the minimum runs over orthogonal matrices and $\Vert \cdot \Vert_F$ is the Frobenius norm. 
    \end{proposition}
    
    \begin{remark}  
        The principle that the spectral properties of the adjacency matrix are matched by a limiting operator also extends to more general models, such as those where the kernel $K$ may depend on the cluster labels in addition to the latent positions. 
        An explicit expression for the spectrum and eigenfunctions of the latter operator will however only be possible in special cases. 
        For instance, the special case where $\cX$ has toroidal geometry and $\mu$ is the uniform measure that was considered in  \cite{avrachenkov2021higher} could be recovered in this fashion. 
    \end{remark}
    \begin{remark}        
        Models where the geometric kernel only depends on the distance between latent variables have been studied in a number of works that aimed to understand the performance of (non-spectral) clustering algorithms when the distance scales with $n$, thus making the geometry ever more pronounced \cite{galhotra2018geometric,abbe2021community,gaudio2024exacta,avrachenkov2024community}.  
        For instance, take $K (X_v, X_w) = k(\gamma_n \Vert X_v - X_w \Vert)$ for some sequence $\gamma_n\to \infty$ and fixed compactly supported function $k:\bbR_{\geq 0}\to [0,1]$.
        It would be interesting future work to establish results similar to Propositions \ref{prop: EigenvaluesMatch} and \ref{prop: EigenvectorsMatch} in such regimes, as this would clarify how strong the geometry can be before spectral methods break. 
    \end{remark}
    \section{Spectral clustering algorithm}\label{sec: SpectralClusteringAlgorithm}
    A generalization of the Perron--Frobenius theorem to integral operators ensures that the greatest eigenvalue of $\mathbb{K}$, denoted $\kappa_*$, has an eigenfunction $\varphi_*\in L^2(\cX,\mu)$ that is strictly positive: 
    \begin{align} 
        \varphi_*(x) >0\ \textnormal{ for all }\ x\in \cX.\label{eq: phi_star_strictpositive} 
    \end{align}
    Specifically, this follows by a variation on Jentzsch' theorem \cite[Theorem I]{horiguchi1996variation}; see Lemma \ref{lem: StrictPos}. 
    All other eigenfunctions of $\mathbb{K}$ have non-constant sign due to their orthogonality. 
    
    The entries of the associated vector $\psi_*^- \in \bbR^n$ defined in \eqref{eq: Def_psi} will then have sign that is aligned with the cluster labels. 
    Hence, by Proposition \ref{prop: EigenvectorsMatch}, we may expect that $\bA$ will have an eigenvector with eigenvalue close to $\lambda_* \de n\rho_n (a-b)\kappa_*/2$ that is informative for clustering. 
    Indeed, if this eigenvalue is simple in the limiting spectrum, then one can deduce that clustering based on the sign of this eigenvector yields almost exact recovery of the cluster labels.  
    Such a simple sign-based clustering algorithm however has a few disadvantages. 
    
    First, sign-based clustering becomes unstable when the limiting eigenvalue has nontrivial multiplicity.
    The limiting eigenfunctions are then not uniquely defined---any linear combination again giving an eigenfunction of the same eigenvalue---which causes the eigenvectors of the adjacency matrix with eigenvalues close to $n\rho_n (a-b)\kappa_*/2$ to approximate linear combinations of the vectors $\psi^{\pm}_j$. 
    This is why the orthogonal transformations occur in \eqref{eq:Procruste}. 
    The sign of the entries of a linear combination of $\psi_*^{-}$ with these other vectors may not always respect the cluster labels, which would make sign-based clustering unstable. 
    Example~\ref{example: Separation} below provides simulations  that confirm this instability. 

    Second, sign-based clustering requires that we know \emph{exactly} what eigenvector to use.
    This can be problematic because, different from what happens in classical settings like the stochastic block model, the correct eigenvalue to use is not necessarily the second largest. 
    Rather, the appropriate eigenvalue depends on the model parameters. 
    We will see in Section \ref{sec: RealWorld} that it is in-fact possible to estimate the eigenvalue location using only the observed network and estimates for cluster-level parameters $a,b$, so the situation is not as bad as may first appear, since detailed prior knowledge of the geometry encoded in $K$, $\mu$ and $\cX$ is not necessary. 
    Still, it is preferable to have an algorithm that can work with less information as one may not even know the cluster-level parameters in practice. 
        
    \subsection{Density-based higher-order clustering}
    Rather than assuming that we know the exact ideal eigenvalue location, suppose that we know some set $\Lambda \subseteq \bbR$ that contains it. 
    For instance, one could take $\Lambda = (-\infty, -\varepsilon] \cup [\varepsilon,\infty)$ with $\varepsilon>0$ sufficiently small. 
    Alternatively, if one has a rough guess $\hat{\lambda}_*$ for the location of the ideal eigenvalue then one could take $\Lambda = [\hat{\lambda}_* - \varepsilon,\hat{\lambda}_* + \varepsilon ]$.

    Our algorithm uses the eigenvectors with eigenvalues in $\Lambda$.
    Considering the coordinates of these vectors then yields an embedding of the vertices in low-dimensional Euclidean space. 
    (See Figure \ref{fig:03} below for a simulated example.) 
    The idea is then that the embeddings of vertices from a fixed cluster should depend continuously on their latent position, while different clusters should be well-separated due to the component of the embeddings aligned with $\psi_*^-$. 
    We can exploit this structure using a density-based clustering algorithms like \texttt{DBSCAN} \cite{ester1996density}. 
    Pseudocode for the resulting procedure is presented in Algorithm \ref{alg:high-density}, and the consistency guarantee is stated in Theorem \ref{thm: AlgGuarantee}.
    The proof is given in Section \ref{sec: ProofAlgGuarantee}.

\begin{theorem}\label{thm: AlgGuarantee}
    Suppose that we are given a Borel set $B$ with $(a-b)\kappa_*/2 \in B$ satisfying the usual conditions that $\nu(\partial B) = 0$ and $0\not\in \overline{B}$. 
    Further, assume that $\rho_n = \omega(\ln(n)/n)$ and let 
    \begin{align} 
        \Lambda_n \de \{ n\rho_n y : y\in B \}. \label{eq: Def_Kn}
    \end{align}
    Then, there exist  $c_1,c_2>0$ such that the partition $\sC$ of $\{1,\ldots,n \}$ output by Algorithm \ref{alg:high-density} with parameters $\varepsilon = c_1/\sqrt{n}$ and $\textnormal{\texttt{MinPts}} = c_2n$  satisfies the following asymptotically almost surely: 
    \begin{description}
        \item[(a)] {\bf Two linearly sized clusters.} All except two parts in $\sC$ have cardinality $o(n)$. The remaining two parts $\cC_1,\cC_2 \in \sC$ satisfy that $\#\cC_i = n/2 - o(n).$ 
        \item[(b)] {\bf Almost exact recovery.} There exists a permutation $\pi$ such that $\sigma_v =1$ for all except $o(n)$ members $v \in \cC_{\pi(1)}$, and such that $\sigma_w = 2$ for all except $o(n)$ members $w \in \cC_{\pi(2)}.$  
    \end{description}
\end{theorem}

\begin{algorithm}[h!]
\fontsize{10.5pt}{12.5pt}\selectfont
\caption{Density-based spectral clustering (\texttt{DBSPEC})}
\label{alg:high-density}
\begin{algorithmic}[1]
\Require The adjacency matrix $\bA \in \{0,1\}^{n\times n}$, a set $\Lambda_n\subseteq \bbR$, and parameters $\epsilon, \texttt{MinPts} > 0$
\Ensure A partition $\sC$ of the vertex set $\{1,\ldots,n\}$

\State Compute eigenvalues $\hat{\lambda}_1, \ldots, \hat{\lambda}_n \in \bbR$ and corresponding eigenvectors $\hat{\psi}_1,\ldots,\hat{\psi}_n \in \bbR^n$ of $\bA$
\State $\cJ_n \gets \{j \leq n : \hat{\lambda}_j \in \Lambda_n\}$
\For{$w = 1,2,\ldots,n$}
    \State $\hat{V}_w \gets [(\hat{\psi}_j)_w : j \in \cJ_n]$
\EndFor
\State $\sC \gets$ the clustering obtained by applying \texttt{DBSCAN} to the points ${\hat{V}_w}$ with parameters $\epsilon$, $\texttt{MinPts}$, and using the Euclidean distance

\end{algorithmic}
\end{algorithm}

\begin{remark}\label{rem: AdmissibleRange}
    The algorithm is not overly sensitive to the choice of the parameters $\epsilon$ and $\texttt{MinPts}$. 
    More precisely, it follows from the analysis that there exists some $C_1 >0$ such that for every fixed $0<c_1< C_1$ every choice of $ \epsilon \in [c_1/\sqrt{n},C_1/\sqrt{n}]$ will work as $n\to \infty$. 
    Moreover, there then exists a $C_2 > 0$ depending on $c_1$ such that for every fixed $c_2>0$ every choice of $\texttt{MinPts} \in [c_2n,C_2n]$ will work as $n\to \infty$.  
\end{remark}

\begin{example}\label{example: Separation}
    Let us illustrate how multi-dimensional spectral embeddings in the new algorithm can be advantageous when the ideal eigenvalue exhibits bad separation from other eigenvalues.
    We revisit the soft geometric block model from Example~\ref{example: Geometric} with $K(x,y)=\bb 1\{\|x-y\|\le R\}$ and let $\cX \de \bbT^1$ be a one-dimensional torus $\bbT^1 \de \{(x,y)\in \mathbb{R}^2: x^2 + y^2 = 1  \}$ (\ie a circle).   
    The eigenvectors in a simulated example for a specific choice of parameters are visualized in Figure \ref{fig:03}. 
    Observe that none of the individual eigenvectors $\psi_i$ close to the ideal eigenvalue produces a clear separation of communities.  
    On the other hand, if the eigenvectors are combined into a $3D$ spectral embedding $V_w=((\psi_2)_w,(\psi_3)_w,(\psi_4)_w)$, then the two communities become cleanly separated. 
    Algorithm~\ref{alg:high-density} would here give accuracy $100\%$, while a sign-based clustering algorithm such as \texttt{HOSC}~\cite{avrachenkov2021higher} using the single eigenvector $\psi_3$ would obtain accuracy only $64\%$.
    
    Additional details related to this example are given in the appendix. We characterize there the conditions under which the ideal eigenvalue is simple.
    In particular, we prove that $ b=a (1-\operatorname{sinc}(2\pi R))/(1+\operatorname{sinc}(2\pi R))$ yields a non-simple eigenvalue in the limit. 
    The qualitative phenomena of this example also serve as a caricature of the expected behavior when the ideal eigenvalue is simple in the limit but has small separation relative to the sampling noise. 
\end{example}

\begin{figure}[htbp]
    \centering
    \begin{minipage}[t]{0.49\textwidth}
        \centering
        \includegraphics[width=0.99\textwidth]{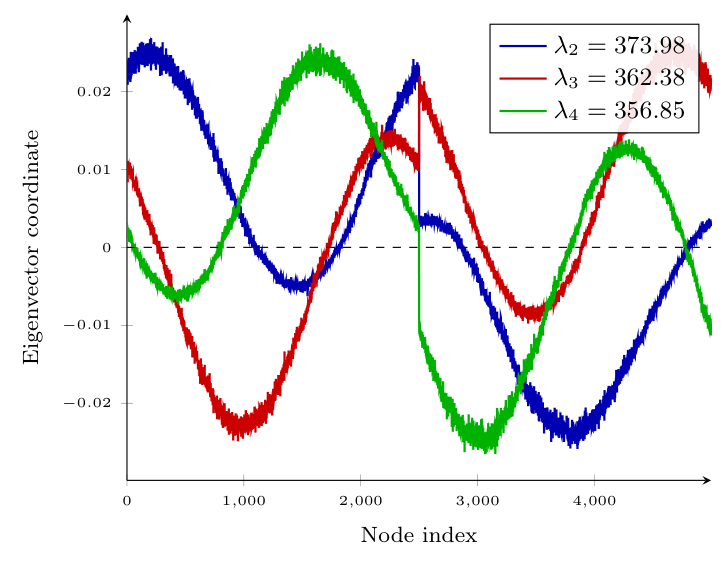} 
    \end{minipage}
    \hfill
    \begin{minipage}[t]{0.49\textwidth}
        \centering
        \includegraphics[width=0.99\textwidth]{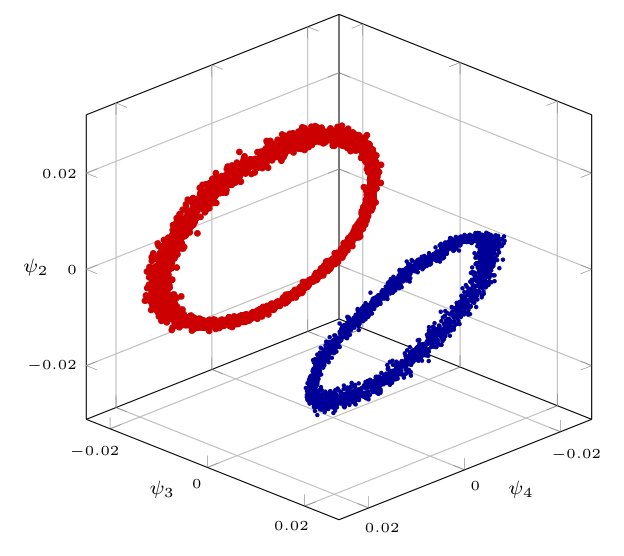} 
    \end{minipage}
    \caption{
    As described in Example \ref{example: Separation}, we consider a simulation of the soft geometric block model on $\cX = \bbT^1$. 
    We take parameters $n=5000$, $R=1/2\pi$, $a=0.5$, and    
    $b=a (1-\operatorname{sinc}(2\pi R))/(1+\operatorname{sinc}(2\pi R))$.  
    It can be shown that the ideal eigenvalue location should here be $362.42$.
    The observed adjacency matrix in-fact has three eigenvalues close to this location: the second largest eigenvalue $\hat{\lambda}_2 \approx 373.98$, third largest $\hat{\lambda}_3 \approx  362.38$ and fourth largest $\hat{\lambda}_4 \approx 356.85$. 
    For comparison, we have $\hat{\lambda}_1 \approx 435.89$ and $\hat{\lambda}_5 \approx 308.20$.   
    (Left) Coordinates of the corresponding eigenvectors $\psi_2$, $\psi_3$, and $\psi_4$, after ordering vertices according to their cluster label and position in $\bbT_1$.
    (Right) Plot of the points $V_w=((\psi_2)_w, (\psi_3)_w,(\psi_4)_w)$ with colors given by the true communities.}
    \label{fig:03}
\end{figure}

\begin{remark}\label{rem: MutlipleClusters}
    Our focus on two clusters is for simplicity, as this case covers all relevant insights with minimal notational burden, but the analysis extends beyond this setting. 
    Suppose that the cluster assignments $\sigma_v$ are in $\{1,\ldots,k \}$ independently with $\mathbb{P}(\sigma_v = j) = q_j$ for $q_1,\ldots,q_k>0$ with $\sum_i q_i = 1$, and that we replace $\bP$ from \eqref{eq: Def_P} by a general symmetric $k\times k$ matrix. 
    Then, using the same argumentation, the limiting operator $\mathbb{A}$ may be replaced by $\tilde{\bP} \otimes \bbK$ with $\tilde{\bP} \de  \bP \operatorname{diag}(q_1,\ldots,q_k) $ implying that the limiting measure in Proposition \ref{prop: EigenvaluesMatch} will be replaced by $\sum_{j=1}^k \sum_{i=1}^\infty \delta_{\tau_j\kappa_i}$ with $\tau_1,\ldots,\tau_k$ the eigenvalues of $\tilde{\bP}$.
    Further, if Assumption \ref{ass: Nontrivial} is generalized to impose that the symmetric matrix $\bP$ has no repeated rows, then a performance guarantee like Theorem \ref{thm: AlgGuarantee} can be derived with the condition $(a-b)\kappa_*/2 \in B$ replaced by $\{\tau_j\kappa_*: j\leq k, \, \tau_j\neq 0 \} \subseteq B$.
    Algorithm \ref{alg:high-density} does not require modification. 
    Arguments for this generalization are given in the appendix.
\end{remark}

    \section{Proof of Propositions \ref{prop: EigenvaluesMatch} and \ref{prop: EigenvectorsMatch}}\label{sec: ProofOperator}     
    
    Denote $\Vert \cdot \Vert$ the spectral norm of a matrix, i.e., $\Vert \bM \Vert = \sup_{x\neq 0} \Vert \bM x \Vert/ \Vert x \Vert$, where $\Vert x \Vert$ denotes the Euclidean norm of $x$. 
    The following preliminary reduction allows us to replace $\bA$ by its conditional expectation given the latent variables $X_v$ and the cluster assignments $\sigma_v$:   
    \begin{lemma}\label{lem: MatrixBernstein}
        Assume that $\rho_n = \omega(\ln(n)/n)$. 
        Then, $\Vert  \bA  - \bbE[\bA  \mid X,\sigma]  \Vert/(n\rho_n) \to 0$ in probability. 
    \end{lemma}
    \begin{proof}
        One can decompose $\bA -\bbE[\bA \mid X,\sigma]$ into a sum of $\binom{n}{2}$ matrices that are conditionally independent given the latent variables $X_v$ and the cluster assignments $\sigma_v$:
        \begin{equation} 
             \bA-\bbE[\bA \vert X,\sigma]  = \sum_{v<w}  \bY(v,w)    \textnormal{ with }  \bigl(\bY(v,w)\bigr)_{i,j} = 
            \begin{cases}
                \bA_{v,w} -\bbE[\bA_{v,w}\vert X,\sigma] & \text{if } \{i,j \} = \{v,w \},\\ 
                0& \text{else}.
            \end{cases}\label{eq:TrustyLion}
        \end{equation}     
        Note that $\Vert \bY(v,w) \Vert \leq 1$. 
        Further, $\Vert \sum_{v<w} \bbE[\bY(v,w)^2 \mid X,\sigma] \Vert \leq n\rho_n$ by a direct computation using that the operator norm of a diagonal matrix is the maximum of its diagonal entries.  
        The matrix Bernstein inequality \cite[Theorem 1.6]{tropp2012user} hence yields that for every $t>0$, 
        \begin{align} 
            \bbP\Bigl(\Vert  \bA  - \bbE[\bA \mid X,\sigma]  \Vert \geq t \Bigr) \leq 2 n \exp\Bigl(\frac{-t^2/2}{n\rho_n  + t/3}\Bigr). \label{eq:BlueTea}
        \end{align}
        Using that $\rho_n = \omega(\ln(n)/n)$ in \eqref{eq:BlueTea} now yields the desired result.  
    \end{proof}
    
    Let $\cX \sqcup \cX \de \{(x,1): x\in \cX  \} \cup \{(x,2): x\in \cX \}$ denote the discrete union of two copies of the latent space $\cX$, and define a measurable function $\sA:(\cX \sqcup \cX) \times (\cX \sqcup \cX) \to [0,1]$ by  
    \begin{align} 
        \sA\bigl((x_1,\sigma_1), (x_2,\sigma_2)) \de  K(x_1,x_2) \bP(\sigma_1,\sigma_2).  
    \end{align}
    Then, the operator $\bbA$ defined in \eqref{eq: Def_bbA} is the integral operator associated with the kernel $\sA$ and the probability measure $\mu \otimes \operatorname{Unif}\{1,2\}$. 
    Further, by definition of the connection probabilities \eqref{eq: Def_bbA},  
    \begin{align} 
        \bbE\bigl[\bA_{v,w} \mid  X,\sigma \bigr]= 
        \begin{cases}
            \rho_n \sA\bigl((X_v, \sigma_v), (X_w, \sigma_w) \bigr) & \textnormal{ if }v\neq w, \\ 
            0 & \textnormal{ if }v=w.
        \end{cases}\label{eq:AvidMan}
     \end{align}
     Propositions \ref{prop: EigenvaluesMatch} and \ref{prop: EigenvectorsMatch} now follow from results by Koltchinskii and Gin\'e \cite{koltchinskii2000random} and Koltchinskii \cite{koltchinskii1998asymptotics} regarding eigenvalues and eigenvectors of matrices sampled from an integral operator: 
     \begin{proof}[Proof of \texorpdfstring{Proposition \ref{prop: EigenvaluesMatch}}{Proposition}]
        Consider the \emph{rearrangement distance} on $\ell^2$ sequences, given by  $
            \delta_2(x,y) \allowbreak \de \inf_{\pi} \sqrt{\sum_{i} (x_i -y_{\pi(i)})^2} 
        $
        where $\pi$ runs over permutations. 
        It is shown in \cite[Theorem 3.1]{koltchinskii2000random} that if $\bH = (\bb1\{v\neq w \} h(X_v,X_w))_{v,w=1}^n$ is a matrix sampled from an arbitrary Hilbert--Schmidt integral operator with symmetric kernel $h$, then the sequence of eigenvalues of $\bH/n$ converges to that of the integral operator with respect to the rearrangement distance.
        Applying this with $h = \sA$ shows that the eigenvalues of $\bbE\bigl[\bA \mid  X,\sigma \bigr]/(n\rho_n)$ converge to those of $\bbA$.      
        In particular, this implies convergence of the empirical counts for sets $B$ with $\nu(\partial B) = 0$ and $0\not\in \overline{B}$ as are considered in Proposition \ref{prop: EigenvaluesMatch}.

        It remains to transfer this result from $\bbE[\bA\mid X,\sigma]$ to the eigenvalues of $\bA/(n\rho_n)$. 
        The assumption that $\nu(\partial B) = 0$ implies that we can find some $\varepsilon$ such that every eigenvalue of $\mathbb{A}$ is at distance $> \varepsilon$ from $\partial B$. 
        The aforementioned convergence in rearrangement distance hence also implies that the eigenvalues $\tilde{\lambda}_1\leq  \cdots\leq \tilde{\lambda}_n$ of $\bbE[\bA \mid X,\sigma]/(n\rho_n)$ have distance $>\varepsilon$ from $\partial B$ for large $n$. 
        It hence suffices to show that the eigenvalues  $\hat{\lambda}_1 \leq\ldots\leq \hat{\lambda}_n$ of $\bA/(n\rho_n)$ satisfy $\vert  \hat{\lambda}_i - \tilde{\lambda}_i \vert \leq \varepsilon$ for all large $n$. 
        Indeed, this follows from 
        Lemma \ref{lem: MatrixBernstein} since Weyl's  perturbation inequality \cite[Theorem 4.3.15]{horn2012matrix} implies that $\vert  \hat{\lambda}_i - \tilde{\lambda}_i \vert \leq \Vert (\bA -\bbE[\bA \mid X,\sigma])/(n\rho_n) \Vert$.
     \end{proof}

    \begin{proof}[Proof of \texorpdfstring{Proposition \ref{prop: EigenvectorsMatch}}{Proposition}]
        This proceeds similarly to the proof of Proposition \ref{prop: EigenvaluesMatch} except that we now use \cite[Theorem 2.1]{koltchinskii1998asymptotics} to control the eigenvectors of $\bbE\bigl[\bA \mid  X,\sigma \bigr]/(n\rho_n)$ and then combine  Lemma \ref{lem: MatrixBernstein} with the Davis--Kahan theorem \cite[Theorem 2]{yu2015useful}. 

        The result of \cite{koltchinskii1998asymptotics} is stated in the terminology of empirical process theory and involves suprema over a function class.
        It may not be immediately obvious that it gives the claimed estimate. 
        To this end, consider the function class $\cF$ given by the finite set consisting of the eigenfunctions $(\varphi_i, \pm \varphi_i)$ of $\bbA$ with eigenvalue $(a\pm b)\kappa_i/2$ in $B$.\footnote{This function class satisfies the conditions of the theorem because we have considered finitely many continuous functions $\varphi_i$ on a compact set. (For instance, this immediately implies that the functions are uniformly bounded.)}
        Then, it follows from \cite[Theorem 2.1, Equation (2.2)]{koltchinskii1998asymptotics} and the orthogonality of eigenfunctions that with $\psi_i^\pm$ as in \eqref{eq: Def_psi}, 
         \begin{align} 
            \max_{i,j}\max_{s,t = \pm} \lvert \langle \tilde{\bP}_B \psi_i^{s}, \psi_j^{t}\rangle - \bb1\{i=j,\,  s =t \}  \rvert \to 0, \label{eq:IvoryPop}
        \end{align}
        almost surely.
        Here, $\tilde{\bP}_B$ is the orthogonal projection on the eigenspace of $\bbE[\bA \mid X,\sigma]/(n\rho_n)$ for the eigenvectors with eigenvalues in $B$. 
        Equivalently, $\Vert \Psi_B^{T}\tilde{\bP}_{B} \Psi_B -\bI \Vert_{F} \to 0$.

        Let $\tilde{\Psi}_B$ have columns given by the eigenvectors of $\bbE[\bA \mid X,\sigma]/(n\rho_n)$ with eigenvalues in $B$. 
        Then, $\tilde{\bP}_B  = \tilde{\Psi}_B\tilde{\Psi}_B^{\T}$ so \eqref{eq:IvoryPop} is equivalent to  $\bQ \de \tilde{\Psi}_B^{\T}\Psi_B$ satisfying $\bQ^{\T}\bQ \to \bI$. 
        Hence, $\min_{\tilde{\bO}} \Vert \bQ - \tilde{\bO} \Vert_{\F}\to 0$ where the minimum runs over orthogonal matrices.  
        Further, using that $\Vert \bM \Vert_{\F}^2 = \Tr(\bM^{\T} \bM)$ as well as the definition of $\bQ$ and the orthogonality of the eigenvectors $\tilde{\psi}_i$, 
        \begin{align} 
            \Vert \Psi_B - \tilde{\Psi}_{\bB}\bQ  \Vert_{\F}^2 &= \Tr\bigl( (\Psi_{\bB} - \tilde{\Psi}_{\bB}\bQ)^{\T} (\Psi_{\bB} - \tilde{\Psi}_{\bB}\bQ) \bigr) \nonumber\\ 
            &= \Tr(   \Psi_B^{\T}\Psi_B - \Psi_{\bB}^{\T} \tilde{\Psi}_{\bB} \bQ -  \bQ^\T \tilde{\Psi}_{\bB}^{\T}\Psi_{\bB} +  \bQ^\T \tilde{\Psi}_{\bB}^{\T} \tilde{\Psi}_{\bB} \bQ  )  =\Tr( \Psi_B^{\T}\Psi_B- \bQ^{\T}\bQ ).  
        \end{align}
        Note that $\Psi_B^{\T}\Psi_B\to \bI$ in probability by using the law of large numbers with the definition \eqref{eq: Def_psi} and the orthogonality of the limiting eigenfunctions. 
        Hence, using that $\bQ^{\T}\bQ \to \bI$ and $\min_{\tilde{\bO}} \Vert \bQ - \tilde{\bO} \Vert_{\F}\to 0$  we have $\min_{\tilde{\bO}}\Vert \Psi_B- \tilde{\Psi}_B\tilde{\bO} \Vert_{\F} \to 0$.

        Further, Proposition \ref{prop: EigenvaluesMatch} and the assumption that $\nu(\partial B)= 0$ implies an asymptotically nonzero spectral gap near $B$.  
        Hence, by Lemma \ref{lem: MatrixBernstein} with the Davis--Kahan theorem \cite[Theorem 2]{yu2015useful}, we have $\min_{{\bO}}\Vert \tilde{\Psi}_B-\hat{\Psi}_B\bO \Vert_{\F} \to 0$ in probability. 
        Thus, $\min_{\hat{\bO}}\Vert \Psi_B- \hat{\Psi}_B\hat{\bO} \Vert_{\F} \to 0$. 
    \end{proof} 
    \section{Proof of \texorpdfstring{Theorem \ref{thm: AlgGuarantee}}{Theorem}}\label{sec: ProofAlgGuarantee}
    We shall show that the situation observed in Figure \ref{fig:03} always occurs: the spectral embeddings concentrate near a connected set for nodes in a fixed cluster, and there is a minimal spacing between the different clusters.    
    It is intuitive that this implies that \texttt{DBSPEC} with well-chosen parameters is consistent, which we will make rigorous.

    \subsection{Preliminaries}
    For easy reference in the subsequent proofs, we here collect some standard facts regarding the \texttt{DBSCAN} algorithm and regarding the eigenfunctions of $\bbK$.
    \subsubsection{Properties of \texttt{DBSCAN}}
    Given a finite set of points $\sP$ in a Euclidean space as well as parameters $\epsilon, \texttt{MinPts}>0$, we recall the following definitions due to the classical paper that introduced \texttt{DBSCAN} \cite[Section 3]{ester1996density}. 

    \begin{definition}\label{def: Neighborhood_Corepoint}
        The \emph{neighborhood} of a point $p\in \sP$ is $N_{\epsilon}(p) \de \{q \in \sP: \Vert p - q \Vert \leq \epsilon  \}$.
    \end{definition}
    \begin{definition}\label{def: Corepoint}
        We say that $p\in \sP$ is a \emph{core point} if $\# N_{\epsilon}(p) \geq \texttt{MinPts}$.  
    \end{definition}
    \begin{definition}\label{def: DensityReach}
        A point $p \in \sP$ is \emph{directly density reachable from $c\in \sP$} if $c$ is a core point and $p\in N_{\epsilon}(c)$.
        Moreover, $p\in \sP$ is \emph{density reachable} from $c$ if there exist $c_1,\ldots,c_n \in \sP$ with $c_1 = c$ and $c_n = p$ such that $c_{i+1}$ is directly density reachable from $c_i$ for every $i$.  
    \end{definition}
    \begin{definition}\label{def: DensityConnected}
        Points $p,q\in \sP$ are \emph{density connected} if there is a point $c\in \sP$ such that both $p$ and $q$ are density reachable from $c$. 
    \end{definition}    
    
    The output of \texttt{DBSCAN} is not entirely deterministic, since the assignment of points $p$ that are density reachable from core points $c_1,c_2\in \cP$ which are not themselves density connected can depend on the order that data was processed; see \cite[p.229--230]{ester1996density}.
    This ambiguity does not affect core points, nor does it affect points that are not density reachable from any core point. 
    Hence, the following guarantee follows from \cite[Definition 6 \& Lemma 2]{ester1996density}:
    \begin{lemma}\label{lem: DBSCAN}
        A partition $\sC$ of $\sP$ output by \texttt{DBSCAN} consists of sets $\cC_1,...,\cC_k, \cC_{\textrm{noise}} \subseteq \sP$ satisfying the following properties: 
        \begin{enumerate}
            \item It holds that $p \in \cC_{\textrm{noise}}$ if and only if $p$ is not density reachable from any core point.\label{item: Noise} 
            \item Core points $c_1,c_2 \in \sP$ belong to the same part if and only if they are density connected. \label{item: CorePoint} 
        \end{enumerate}
    \end{lemma}
    Note that item \ref{item: Noise} together with the fact that $\sC$ is a partition implies that every point that is density reachable from a core point belongs to some $\cC_i$ with $i\leq k$. 
    In particular, this is the case for the core points themselves.  
    
    \subsubsection{Basic properties of the eigenfunctions}
    \begin{lemma}\label{lem: Continuity}
        Every eigenfunction associated with a non-zero eigenvalue of $\mathbb{K}$ has a continuous representative. 
    \end{lemma}
    \begin{proof}
        Let $\varphi$ be eigenfunction with eigenvalue $\kappa$, and suppose that $\varphi$ is normalized to have unit norm in $L^2(\cX, \mu)$.
        Then, by the eigenvector equation, $\varphi(x) - \varphi(y) = \frac{1}{\kappa}\int_{\cX} \varphi(z) (K(x,z) - K(y,z) \intd \mu(z)$. 
        Hence, by Cauchy--Schwarz, $\lvert\varphi(x) - \varphi(y)  \rvert \leq \frac{1}{\lvert \kappa \rvert} \Vert K(x,\cdot) - K(y,\cdot) \Vert_{L^2(\cX,\mu)}$.
        By Assumption~\ref{ass: ContinuityKernel}, the desired result follows.
     \end{proof}
    \begin{lemma}\label{lem: StrictPos}
        The eigenfunction $\varphi_*$ for the greatest eigenvalue of $\mathbb{K}$ is strictly positive.  
    \end{lemma}
    \begin{proof}
        As was announced just after \eqref{eq: phi_star_strictpositive}, the statement follows from \cite[Theorem I]{horiguchi1996variation}.
        It may, however, not be obvious that that theorem is applicable.     
        First, let us note that Theorem~I in \cite{horiguchi1996variation} is stated for the Lebesgue measure on a subset of $\bbR^n$. 
        This is however merely a matter of formulation: $(\cX,\mu)$ can be realized using a measurable map from the unit interval with the Lebesgue measure, as it is a standard probability space. 
        More significantly, one of the conditions in \cite{horiguchi1996variation} asks that there should be some $\ell \geq 1$ such that the kernel of the integral operator $\bbK^\ell$ is strictly positive almost everywhere.
        We shall show that this is the case. 

        A direct computation shows that $\bbH \de \bbK^2$ has kernel $h(x,y) = \int_\cX K(x,u) K(u,y)\intd \mu(u)$.
        Assumption \ref{ass: Nontrivial} then implies that $h(x,x)>0$ for every $x\in \cX$. 
        Further, the continuity of $x\mapsto K(x,\cdot)$ from Assumption \ref{ass: ContinuityKernel} implies that $h$ is continuous as a bivariate function. 
        In particular, every $x\in \cX$ admits some $\varepsilon_x>0$ and open neighborhood $\cU_x$ such that $h(y,z)>\varepsilon_x$ for every $y,z \in \cU_x$. 
        Using that $\cX$ is compact, one can reduce the open cover $\{\cU_x: x\in \cX\}$ to a finite cover $\{\cU_{x_1},...,\cU_{x_k}\}$ of, say, $k\geq 1$ open sets. 
        The connectivity of $\cX$ from Assumption \ref{ass: ConnectedLatent} implies that every $y,z \in \cX$ admit a sequence $i(1),\ldots,i(k)$ satisfying $\cU_{x_{i(j)}}\cap \cU_{x_{i(j+1)}} \neq \emptyset$ such that $y\in \cU_{i(1)}$ and $z \in \cU_{i(k)}$.  
        Using this and that $\mu(\cU_{x_{i(j)}}\cap \cU_{x_{i(j+1)}}) >0$ by Assumption \ref{ass: ConnectedLatent}, one can verify that $\bbK^{2k} = \bbH^k$ indeed has a strictly positive kernel. 
        This concludes the proof.     
    \end{proof}
    \subsection{Proof of the consistency guarantee}
    This section adopts the notation and assumptions from Theorem \ref{thm: AlgGuarantee}. 
    In particular, we assume that $B\subseteq \bbR$ is such that $(a-b)\kappa_* /2\in B$ with the usual conditions $\nu(\partial B) = 0$ and $0\not\in \overline{B}$, and we let $\Lambda_n = \{n\rho_n y: y \in B \}$ as in \eqref{eq: Def_Kn}.  
    
    Recall from Algorithm \ref{alg:high-density} that $\hat{V}_w = [(\hat{\psi}_i)_w :  \hat{\lambda}_i \in \Lambda_n]$ is found by taking the $w$th entry of each eigenvector of $\bA$ with an eigenvalue in $\Lambda_n$. 
    Let us similarly define $V_w \de [(\psi_i^{\pm})_w : (a\pm b)  \kappa_i/2 \in  B]$.
    We use Lemmas \ref{lem: Continuity} and \ref{lem: StrictPos} to establish properties of the vectors $V_w$: 
    
    \begin{lemma}\label{lem: continuity}
        For every $\varepsilon>0$ there exists some $\delta>0$ such that asymptotically almost surely $\Vert V_v - V_w \Vert < \varepsilon/\sqrt{n}$ for all values of $v,w$ with $\Vert X_v -X_w \Vert < \delta$ and $\sigma_v = \sigma_w$. 
    \end{lemma}
    \begin{proof}
        Note that the $V_w$ are $\nu(B)$-dimensional vectors, and that $\nu(B)<\infty$ since there are only finitely many $i$ with $(a+b)\kappa_i/2\in B$ or $(a-b)\kappa_i/2 \in B$ because the only accumulation point of the $\kappa_i$ is at zero and we assume that $0 \not\in \overline{B}$. 
        In other words, the $V_w$ have a fixed finite dimensionality.
        Recall further that $\psi_i^{\pm}$ in the definition of $V_w$ are defined using the eigenfunctions $\varphi_i$ in \eqref{eq: Def_psi} with some proportionality constant (say $Z_n$) chosen to get unit vectors. 
        By the law of large numbers, the proportionality constant satisfies $\sqrt{n}Z_n\to 1$ almost surely.  
        The result hence follows from the continuity of the $\varphi_i$ from Lemma \ref{lem: Continuity}. 
    \end{proof}
    \begin{lemma}\label{lem: discontinuity}
        There exists some $\varepsilon_0>0$ such that such that asymptotically almost surely  $\Vert V_v -V_w \Vert > \varepsilon_0/\sqrt{n}$ whenever $\sigma_v \neq \sigma_w$. 
    \end{lemma}
    \begin{proof}
        Recall from Lemma \ref{lem: StrictPos} that the eigenfunction $\varphi_*$ associated with the greatest eigenvalue $\kappa_*$ of $\bbK$ is strictly positive. 
        In particular, there exists some $c>0$ such that $\varphi_*(x)>c$ for all $x\in \cX$. 
        Recalling that we assumed that $(a-b)\kappa_*/2 \in B$, it follows from the definition \eqref{eq: Def_psi} of $\psi_*^-$ that with $Z_n$ the proportionality constant, 
        \begin{align}
            \Vert V_v - V_w \Vert \geq \lvert (\psi_*^-)_v - (\psi_*^-)_w\rvert  = Z_n\lvert \varphi_*(X_v) + \varphi_*(X_w) \rvert  > 2cZ_n.    
        \end{align}
        Here, we used that $\sigma_v \neq \sigma_w$ in the equality. 
        Use that $\sqrt{n}Z_n\to 1$ almost surely by the law of large numbers and take $\varepsilon_0 = 2c$ to conclude. 
    \end{proof}
    We now determine an admissible range for the parameters in Algorithm \ref{alg:high-density}. 
    Let us fix $c_1,C_1 >0$ such that $c_1<C_1< \varepsilon_0/3$ with $\varepsilon_0$ as in Lemma \ref{lem: discontinuity}. 
    We assume that 
    \begin{align}
        \epsilon \in [c_1/\sqrt{n} , C_1/\sqrt{n}]. \label{eq: Def_eps}
    \end{align}
    Further, let $\delta$ be the value resulting from Lemma \ref{lem: continuity} with $\varepsilon = c_1/2$. 
    Using the compactness of $\cX$ we can find a finite open cover $\{U_j:j=1,\ldots , r\}$ consisting of sets of diameter $<\delta$.
    Recalling from Assumption~\ref{ass: ConnectedLatent} that the support of $\mu$ covers $\cX$, each of these open sets has nontrivial measure. 
    Let us fix $c_2,C_2>0$ with $c_2 < C_2 < \min_{j\leq r} \mu(U_j)/2$, and assume that 
    \begin{align}
        \texttt{MinPts} \in [c_2n, C_2n ]. \label{eq: Def_MinPts}
    \end{align}
    To prove Theorem \ref{thm: AlgGuarantee} it then suffices to show that Algorithm \ref{alg:high-density} achieves almost exact recovery whenever the parameters are in this range.
    Note that this also yields the claim in Remark \ref{rem: AdmissibleRange}. 

    To this end, we shall next transfer the properties of the vectors $V_w$ to the $\hat{V}_w$. 
    Proposition \ref{prop: EigenvectorsMatch} ensures that the vectors $\hat{V}_w$ have dimensionality $\nu(B)$ asymptotically almost surely.
    Hence, we can assume by discarding an event of negligible probability that the dimensions of $V_w$ and $\hat{V}_w$ match.
    It is then also well-defined to consider some $\nu(B)\times \nu(B)$ orthogonal matrix $\hat{\bO}$ with 
    \begin{align}
        \sum_{w=1}^n\Vert \hat{V}_w^\T \hat{\bO} - V_w^\T \Vert^2  = \min_{\bQ \in \bbO(\nu(B))} \sum_{w=1}^n\Vert \hat{V}_w^\T \bQ - V_w^\T \Vert^2.  \label{eq: Def_Ohat}
    \end{align}
    Given such a matrix $\hat{\bO}$, let us define a set of \emph{good indices} by 
    \begin{align}
        \cI_{\textrm{good}} \de \bigl\{w\leq n : \Vert \hat{V}_w \hat{\bO} - V_w \Vert < \gamma/\sqrt{n}  \bigr\} \, \textrm{ where }\, \gamma  \de \min\{c_1/4, \varepsilon_0/3,\varepsilon_0/3 - C_1 \}.\label{eq: Def_Igood} 
    \end{align}
    Almost all indices are good: 
    \begin{lemma}\label{lem: allgood}
        It holds that $\#\cI_{\textrm{good}}/n\to 1$ in probability as $n\to \infty$.  
    \end{lemma}
    \begin{proof}
        Since the square of the Frobenius norm of a matrix is the sum of the squares of the Euclidean norms of its rows, it follows from Proposition \ref{prop: EigenvectorsMatch} and \eqref{eq: Def_Ohat} that 
        \begin{align}
            \sum_{w=1}^n\Vert \hat{V}_w^\T \hat{\bO} - V_w^\T \Vert^2 \to 0\label{eq:MadCat}
        \end{align}
        in probability as $n\to \infty$.
        On the other hand, 
        \begin{align}
           \sum_{w=1}^n\Vert \hat{V}_w^\T \hat{\bO} - V_w^\T \Vert^2 \geq  \gamma^2 \bigl( 1 -\#\cI_{\textrm{good}}/n \bigr).  \label{eq:MadDog}
        \end{align}
        Combine \eqref{eq:MadCat} and \eqref{eq:MadDog} to conclude.
    \end{proof}
    We shall now work towards application of Lemma \ref{lem: DBSCAN} to show that Algorithm \ref{alg:high-density} correctly clusters all good indices. 
    Let us take $\sP \de \{\hat{V}_w : w\leq n \}$. Recall Definitions \ref{def: Neighborhood_Corepoint} to \ref{def: DensityConnected}.
    \begin{lemma}\label{lem: GoodCore}
        Asymptotically almost surely, every $\hat{V}_w$ with $w\in \cI_{\textrm{good}}$ is a core point. 
    \end{lemma}
    \begin{proof}
        For every $w\in \cI_{\text{good}}$, using \eqref{eq: Def_eps} and \eqref{eq: Def_Igood} with the triangle inequality,
        \begin{align}
             \# N_{\epsilon}(\hat{V}_w) = \#\{v\leq n: \Vert \hat{V}_w - \hat{V}_v \Vert \leq \epsilon\} \geq \#\{k\in \cI_{\textrm{good}}: \Vert V_w - V_v \Vert \leq  (c_1 -  2\gamma) /\sqrt{n}  \}. \label{eq:CrazyCat} 
        \end{align}
        Recall from before \eqref{eq: Def_MinPts} that $\{U_j:j\leq r\}$ is an open cover of $\cX$ consisting of sets of diameter $<\delta$, where $\delta$ was chosen to allow an application of Lemma \ref{lem: continuity} with $\varepsilon = c_1/2$. 
        Let $J(w)$ be the random index with $X_w \in U_{J(w)}$.
        Then, 
        $\Vert V_w - V_v \Vert \leq c_1/2\sqrt{n}$ 
        for every $v$ with $X_v \in U_{J(w)}$ and $\sigma_v = \sigma_w$.    
        Hence, using that $\gamma \leq c_1/4$ by \eqref{eq: Def_Igood}, it holds for every $w\in \cI_{\textrm{good}}$ that 
        \begin{align}
             \# N_{\epsilon}(\hat{V}_w) \geq \#\{v\in \cI_{\textrm{good}}: X_v \in U_{J(w)}  \textrm{ and } \sigma_v = \sigma_w\}.\label{eq:CrazyDog} 
        \end{align}
        Combining the weak law of large numbers with Lemma \ref{lem: allgood} yields that 
        \begin{align}
            \max_{j\leq r, s\in \{1,2\}} \bigl\lvert \#\{v\in \cI_{\textrm{good}}: X_v \in U_j \textrm{ and } \sigma_v =s \}/n - \mu(U_j)/2\bigr\rvert \to 0
        \end{align}
        in probability as $n\to \infty$. 
        In particular, recalling from \eqref{eq: Def_MinPts} that $\texttt{MinPts} \leq C_2n$ with $C_2 < \min_{r\leq 2}\mu(U_j)/2$, it indeed holds that $\#N_\epsilon(\hat{V}_w) \geq \texttt{MinPts}$ for every $w\in \cI_{\text{good}}$ asymptotically almost surely. 
        This means that every $V_w$ with $w \in \cI_{\textrm{good}}$ is a core point, as desired.  
     \end{proof}
     \begin{lemma}\label{lem: DenseConnected}
        Fix some $s\in \{1,2\}$.
        Then, asymptotically almost surely, all $\hat{V}_w$ with $w\in \cI_{\textrm{good}}$ and $\sigma_w = s$ are density connected. 
     \end{lemma}
     \begin{proof}
         The  argumentation in \eqref{eq:CrazyCat}--\eqref{eq:CrazyDog} shows that $\Vert \hat{V}_{w} - \hat{V}_{w'} \Vert \leq \epsilon$ for every $w,w'\in \cI_{\textrm{good}}$ with $X_{w}$ and $X_{w'}$ in the same member of the cover $\{U_j:j\leq r\}$ and $\sigma_{w} = \sigma_{w'}$. 
         Considering Lemma \ref{lem: GoodCore} and Definition \ref{def: DensityConnected}, it now suffices to show that for every $w_{A},w_{B} \in \cI_{\textrm{good}}$ with $\sigma_{w_A} =s = \sigma_{w_B}$ there exist $w(0),...,w(k) \in \cI_{\textrm{good}}$ with $w(0) = w_A$ and $w(k) = w_B$, such that $\sigma_{w(m)}=s$ and such that $X_{w(m)}$ and $X_{w(m+1)}$ lie in a shared member of the cover for every $m$.

        Let $U_A$ and $U_B$ be the random members of the open cover $\{U_j:j\leq r\}$ containing $X_{w_A}$ and $X_{w_B}$, respectively. 
        The connectivity of $\cX$ implies that there exist  $j(1),...,j(\ell)$ with $U_{j(1)} = U_A$ and $  U_{j(\ell)} = U_{B}$ such that $U_{j(m)} \cap U_{j(m+1)}\neq \emptyset$ for every $1\leq m < \ell$. 
        Here, the weak law of large numbers together with Lemma \ref{lem: allgood} implies that 
        \begin{align}
            \min_{U_{j}\cap U_{j'} \neq \emptyset} \#\{w\in \cI_{\textrm{good}}: X_{w}\in U_{j}\cap U_{j'} \text{ and } \sigma_w = s \}/n \to \min_{U_{j} \cap U_{j'} \neq \emptyset} \mu(U_j \cap U_{j'})/2 \label{eq:BatSCat}   
        \end{align}
        in probability. 
        This limit is strictly positive since the support of $\mu$ covers $\cX$.
        In particular, all the sets on the left-hand side of \eqref{eq:BatSCat} are nonempty asymptotically almost surely. 
        Thus, there exist $w(1),\ldots, w(\ell)\in \cI_{\textrm{good}}$ with $X_{w(m)} \in U_{j(m)}\cap U_{j(m+1)}$ and $\sigma(w(m))=s$ for every $m$.  
        Let $k=\ell +1$ and set $w(0) = w_A$ and $w(k) = w_B$ to conclude. 
        \end{proof}

        \begin{lemma}\label{lem: NotDense}
            Asymptotically almost surely, it holds for every $w_1,w_2 \in \cI_{\textrm{good}}$ with $\sigma_{w_1} = 1$ and $\sigma_{w_2} = 2$ that $\hat{V}_{w_1}$ and $\hat{V}_{w_2}$ are not density connected.    
        \end{lemma}
        \begin{proof}
            For every $s \in \{1,2\}$ define a subset $\cR_s \subseteq \bbR^{\nu(B)}$ by 
            \begin{align}
                \cR_s \de \{x \in \bbR^{\nu(B)}:   \Vert x^\T\hat{\bO}  -  V_w^\T\Vert < \varepsilon_0/3\sqrt{n} \textrm{ for some }w\leq n \textrm{ with }\sigma_w =s  \}.
            \end{align}
            Lemma \ref{lem: discontinuity} ensures that $\Vert V_{w_1} - V_{w_2}\Vert > \varepsilon_0/\sqrt{n} $ for every $w_1,w_2$ with $\sigma_{w_1} =1$ and $\sigma_{w_2} =2$. 
            Hence, recalling from \eqref{eq: Def_eps} that $\epsilon \leq C_1/\sqrt{n} < \varepsilon_0/3\sqrt{n}$ every point in $\cR_1$ is at distance $>\epsilon$ from every point in $\cR_2$. 
            Further, Definition \ref{def: DensityConnected} yields that two core points $c_A,c_B$ are density connected if and only if one can find a sequence of core points $c_1,...,c_k$ with step sizes $\Vert c_{j+1} - c_j \Vert \leq \epsilon$ such that $c_1 = c_A$ and $c_k = c_B$.
            It now suffices to show for every $s$ that every $\hat{V}_{w_s}$ with $w_s \in \cI_{\textrm{good}}$ and $\sigma_{w_s} = s$ satisfies $\hat{V}_{w_s} \in \cR_s$, and that every $\hat{V}_w$ with $\hat{V}_w \not\in \cR_1 \cup \cR_2$ is not a core point. 

            Regarding the first statement, the definition $\cI_{\textrm{good}}$ yields that $\Vert \hat{V}_{w_s}^{\T}\hat{\bO} - V_{w_s}^\T\Vert < \gamma/\sqrt{n}$ for every $w_s\in \cI_{\textrm{good}}$. 
            So, indeed $\hat{V}_{w_s}\in  \cR_s$ since $\gamma \leq \varepsilon_0/3$; recall \eqref{eq: Def_Igood}. 
            Regarding the second statement, it is immediate from the definition of $\cR_s$ that for every $\hat{V}_w\not\in \cR_1\cup \cR_2$,  
            \begin{align}
                \#N_\epsilon(\hat{V}_w) \leq \bigl\{ v\leq n:  \Vert \hat{V}_{v}^{\T}\hat{\bO} - V_{v}^{\T}\Vert \geq   \varepsilon_0/3\sqrt{n}-\epsilon \bigr\}.
            \end{align} 
            Recall from \eqref{eq: Def_eps} that $\epsilon\leq C_1/\sqrt{n}$ and recall from \eqref{eq: Def_Igood} that $\gamma \leq  \varepsilon_0/3 - C_1$. 
            It hence follows that  
            \begin{align}
                \bigl\{ v\leq n:  \Vert \hat{V}_{v}^{\T}\hat{\bO} - V_{v}^{\T}\Vert \geq   \varepsilon_0/3\sqrt{n}-\epsilon \bigr\} \leq n -  \#\cI_{\textrm{good}}.
            \end{align}
            Lemma \ref{lem: allgood} and the assumption that $\texttt{MinPts} \geq c_2n$ from \eqref{eq: Def_MinPts} now ensures that asymptotically almost surely $\#N_\epsilon(\hat{V}_w) < \texttt{MinPts}$ for every $\hat{V}_w\not\in \cR_1\cup \cR_2$.
            This means that $\hat{V}_w$ is not a core point, concluding the proof. 
        \end{proof}
        \begin{proof}[Proof of Theorem \ref{thm: AlgGuarantee}]
            Lemmas \ref{lem: GoodCore} and \ref{lem: DenseConnected} imply that for every $s\in \{1,2 \}$ the set $\cV_i \de \{\hat{V}_w: w\in  \cI_{\textnormal{good}}, \sigma_w = s \}$ consists of mutually density connected core points. 
            Moreover, Lemma \ref{lem: NotDense} yields that $\cV_1$ is not density connected to $\cV_2$.      
            Lemma \ref{lem: DBSCAN} hence implies that there exist two distinct parts $\cC_1 \neq \cC_2$ in the partition output by \texttt{DBSCAN} with $\cV_1 \subseteq \cC_1$ and $\cV_2 \subseteq \cC_2$. 
            Using that $\#\cI_{\textnormal{good}} = n + o(n)$ by Lemma \ref{lem: allgood} now concludes the proof. 
        \end{proof}
    \section{Experiments with real-world datasets}\label{sec: RealWorld} 
    Recall that some experiments with synthetic data were given in Example \ref{example: Separation} and more extensively in the appendices. 
    We here consider three real-world network datasets commonly used in the community detection literature:  
\begin{itemize}
    \item \textbf{LiveJournal \cite{livejournal}.}  
    A social network extracted from the LiveJournal online blogging platform, with edges corresponding to declared friendship relations. The dataset was released as part of the SNAP collection at Stanford University and is characterized by a heterogeneous degree distribution and user-defined community memberships.
    \item \textbf{Political Blogs \cite{polblogs}.} A hyperlink network among U.S. political blogs. The ground-truth partition is given by political leaning (liberal vs. conservative), making this dataset an example of a two-community structure.
    \item \textbf{DBLP \cite{SNAPDBLP}.}  
    A collaboration network derived from the DBLP computer science bibliography. 
    Nodes represent authors and edges indicate co-authorship relationships. 
    Ground-truth community labels are based on publication venues. 
    We consider two classification examples: DBLP A (top communities 1 and 2) and DBLP B (top communities 1 and 3). 
    We will also present results in a derived dataset,  called DBLP A$^\ast$, obtained by removing the nodes from the 2 biggest cliques of DBLP A. 
\end{itemize}

In each dataset considered, we removed overlapping nodes (nodes that are in the two communities) and isolated nodes. 
Table~\ref{tab:real_datasets} summarizes the main structural properties of each dataset, including the number of nodes ($n$), average degree ($\bar{d}$) and median degree ($\Tilde{d}$), community sizes  ($n_1,n_2$), and the cluster connectivity parameters $(a_1,a_2)$ and $(b)$. 
Here, $a_i$ is the internal connection density of the $i$th cluster, defined as the ratio of the number of edges connecting two nodes in that cluster to the total number of possible edges.
Similarly, $b$ is the inter cluster density defined in terms of the edges between clusters.

\begin{table}[h]
\centering
\caption{Structural properties of the real-world datasets.}
\label{tab:real_datasets}
\begin{tabular}{lcccccccc}
\hline
Dataset & $n$ & $\bar{d}$  & $\Tilde{d}$ & $n_1$ & $n_2$ & $a_1$ & $a_2$ & $b$ \\
\hline
Political Blogs & 1224 & 27.31 & 13 &588& 636 & 0.042 & 0.039 & 0.0042 \\	
DBLP A & 13036 & 4.94 & 3 & 7321& 5715 & 0.000629 & 0.000876 & 0.000024 \\
DBLP A$^\ast$ & 12991 & 4.87 & 3 & 7297& 5694 & 0.000621 & 0.000872 & 0.000023\\
DBLP B & 12212 & 6.10 & 4 & 7566 & 4646 & 0.000635 & 0.0017 & 0.00002 \\
LiveJournal  & 2766 & 17.45 & 7 & 1426 & 1340 & 0.0061 & 0.0197 & 0.00014 \\
\hline
\end{tabular}
\end{table}

One testable prediction from our model is on the location of the ideal eigenvalue for the purpose of clustering. 
In the setting of our main discussion where the two clusters were assumed to have the same size and internal connection probabilities, our results predicted that this ideal eigenvalue should be located near $\lambda_* = n\rho_n (a-b)\kappa_*/2$ with $\kappa_*$ the eigenvalue of the kernel operator $\bbK$.
This location depends on the latent geometry through $\kappa_*$, which is information that we do not have direct access to, but this dependence can be eliminated through an observable quantity. 
Indeed, note that Proposition \ref{prop: EigenvaluesMatch} predicts that the largest eigenvalue of the adjacency matrix should be $\hat{\lambda}_{\textnormal{max}} \approx n\rho_n (a+b)\kappa_*/2$ and this is an observable quantity. 
Thus, the ideal eigenvalue should be located near $\hat{\lambda}_{\textnormal{max}}(a-b)/(a+b)$.

Inspecting Table \ref{tab:real_datasets}, we see that the assumption of equal community sizes and internal connection densities is violated in some of the considered datasets. 
To account for this, we consider the variants of our results described in Remark~\ref{rem: MutlipleClusters}.
Consider the matrix
\begin{equation} 
        \tilde{\bP} \de \begin{pmatrix}
            n_1a_1/n & n_2b/n\\
            n_1b/n & n_2a_2/n
        \end{pmatrix}.\label{eq: P_real}
\end{equation} 
Let $\tau_{\min}$ and $\tau_{\max}$ denote the minimum and maximum eigenvalues of $\Tilde{\bP}$, respectively. 
Then, by Remark~\ref{rem: MutlipleClusters}, the ideal eigenvalue for clustering should be located near the following value:  
\begin{equation}
    \hat{\lambda}_*\de  \hat{\lambda}_{\max}\frac{\tau_{\min}}{\tau_{\max}}.\label{eq:GrumpySax}
\end{equation}

To test the quality of eigenvalue locations we consider spectral clustering based on the sign of the entries of a single eigenvector.
We do this both for the eigenvector for the second eigenvalue (classical spectral clustering) and for an eigenvector with a higher-order eigenvalue (HO). 
In the higher-order case, we consider both the eigenvalue closest to $\hat{\lambda}_*$ and the eigenvalue $\hat{\lambda}_{\textnormal{opt}}$, which was found to yield the best performance. Table~\ref{tab:real_results} reports the accuracy of the resulting clusterings, measured as the fraction of correctly classified nodes (up to label permutation), together with the indices of the eigenvalues selected by the higher-order methods.

\begin{table}[h]
\centering
\caption{Sign-based clustering accuracy on real-world datasets.}
\label{tab:real_results}
\begin{tabular}{lccccc}
\hline
Dataset & Classical & HO $(\hat{\lambda}_*)$ & HO $(\hat{\lambda}_{\textnormal{opt}})$ &Index $\hat{\lambda}_*$ & Index $\hat{\lambda}_{\textnormal{opt}}$\\
\hline
Political Blogs &93\% & 93\%  & 93\% & 2& 2\\
DBLP A & 56\% & 56\% & 76\% & 1 & 12\\
DBLP A$^\ast$& 56\% & 76\% & 76\% & 10& 10\\
DBLP B & 65\% & 74\% & 75\% &4 & 15\\
LiveJournal & 56\% & 77\%  & 85\% & 3 & 4\\
\hline
\end{tabular}
\end{table}

The classical and higher-order methods naturally have identical performance when the corresponding eigenvalues coincide, as in Political Blogs. However, the second eigenvalue does not always match the one predicted by our theory.
When it does not, selecting the predicted higher-order eigenvalue can substantially improve performance.
For instance, for LiveJournal,  selecting the third eigenvector improves the accuracy from the almost random $56\%$ to a respectable $77\%$. 
For DBLP B, accuracy is improved from $65\%$ to $74\%$ by selecting the fourth eigenvector. 
The eigenvectors indexed by $9$ and $15$ here achieve almost the same accuracy as the $4$-th one, and this is the greatest accuracy obtained in any eigenvector.

We now discuss DBLP A in more detail. Figure~\ref{fig:dblp_comparison} (left) shows that the best accuracy is achieved by the 12th eigenvector, associated with  $\hat{\lambda}_{12}=14.64$. 
The eigenvalues closest to $\hat{\lambda}_*=20.78$ are however $\hat{\lambda}_1=23.12$ (the closest) and $\hat{\lambda}_2=17$.
To explain this discrepancy, let us note that this dataset has one clique of size 18 and another of size 22. 
A clique of size $n$ has a large leading eigenvalue $n-1$.  
By the interlacing theorem, the largest eigenvalue of DBLP A must therefore be greater than $21$. 
Since our estimator uses the largest eigenvalue of the graph, these large cliques may significantly affect the estimate. 

We repeat the experiment after removing these two cliques, thereby obtaining the modified dataset DBLP A$^\ast$.
Now, the two largest eigenvalues of the adjacency matrix are  $\hat{\lambda}_1=16$ and  $\hat{\lambda}_2=15.67$. 
Corresponding to the significantly changed value for $\hat{\lambda}_1 = \hat{\lambda}_{\textnormal{max}}$, our estimator $\hat{\lambda}_*=14.31$ also changes substantially and now accurately reflects the optimal value; see Figure~\ref{fig:dblp_comparison} (right). 
While removing the clique does not substantially change the best possible achievable accuracy, it does have significant effect on the estimator for the optimal eigenvalue location.

\begin{figure}[htbp]
    \centering
    \begin{minipage}[t]{0.495\textwidth}
        \centering
        \includegraphics[width=0.97\textwidth]{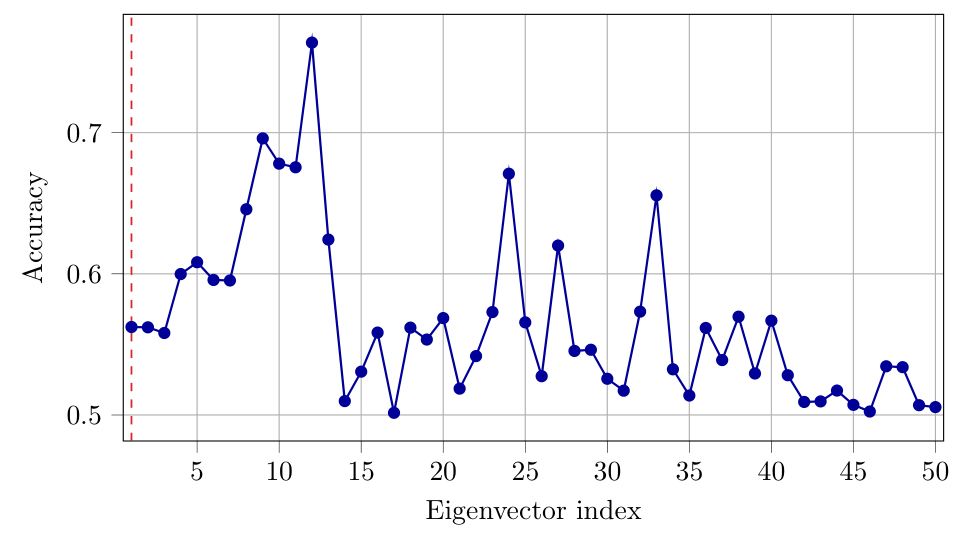}
    \end{minipage}
    \hfill
    \begin{minipage}[t]{0.485\textwidth}
        \centering
        \includegraphics[width=0.95\textwidth]{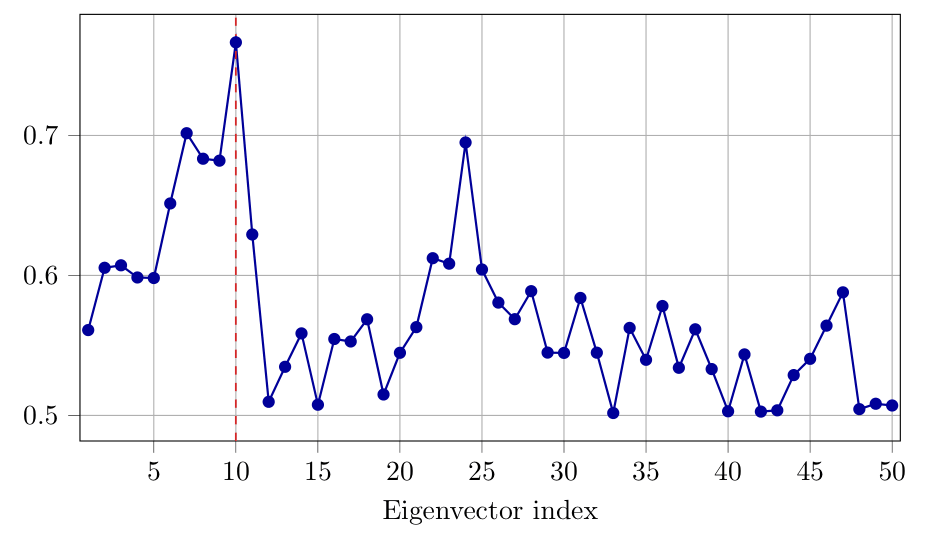}
    \end{minipage}
    \caption{
    Accuracy of the sign assignment for varying eigenvector indices.
    The left and right figures show DBLP A and A$^*$, respectively. 
    The dashed red line indicates the index of the eigenvalue closest to $\hat{\lambda}_*$. 
    }
    \label{fig:dblp_comparison}
\end{figure}

The foregoing concerns sign-based clustering with a single eigenvector. 
We next illustrate the \texttt{DBSCAN}-based algorithm \texttt{DBSPEC} with several eigenvectors on  LiveJournal.\footnote{Discussion of \texttt{DBSPEC} clustering for DBLP A$^\ast$ is deferred to the appendices for brevity. 
We there find that performance is also improved in DBLP A$^\ast$, and surprisingly find signals that there are more than $2$ clusters.  
}

Recall that higher-order sign-based clustering improved performance in many of our datasets. 
In applications, however, one may not know the parameters $a_1,a_2,b$ nor $n_1,n_2$ needed to determine $\hat{\lambda}_*$.
Less information is needed for density-based spectral clustering with a few leading eigenvectors.  
It is hence natural to consider spectral embeddings by the rows of $[\psi_1\ \psi_2\ \psi_3]$, the eigenvectors of the three largest eigenvalues.  

For LiveJournal, we find that applying \texttt{DBSCAN} on the $3$-dimensional spectral embeddings with an ad-hoc optimized parameter choice $\epsilon = 0.0003$ and $\texttt{MinPts} = 30$ achieves $88\%$ accuracy on the assigned points but labels $62\%$ of the vertices as noise. The accuracy of the clustered points outcompetes sign-based clustering, but a large number of points have been assigned to noise. 
One could address this, and likely achieve fairly high total accuracy, by using some second algorithm to assign the noise points based on the assigned ones (for example, in the spirit of  \cite[Algorithm 2]{avrachenkov2021higher}). However, we will next see that a closer inspection of the spectral embeddings reveals that the noise points have an understandable root cause.   

Observe in Figure \ref{fig:livejournal_leading3_pair} (left) that the spectral embeddings exhibit radial streaks.
These radial streaks cause bad separation near the origin while also causing greater spread of points that are distant from the origin. 
This explains why a small value of $\epsilon = 0.0003$ was necessary for \texttt{DBSCAN} to identify two clusters, but then also results in many points being assigned to noise.  

Such radial streaks are a well-known phenomenon associated with large node degree fluctuations that are common in real-world datasets, and multiple ways to address this phenomenon have been proposed \cite{qin2013regularized,pmlr-v23-chaudhuri12,dasgupta2004spectral,gao2018community}. 
One such method is to eliminate the radial component by rescaling embeddings to the unit sphere, an algorithmic step that dates back to \cite{ng2001spectral}. 

Figure \ref{fig:livejournal_leading3_pair} (right) visualizes the renormalized spectral embeddings. 
We observe a much cleaner separation between the clusters, aside from only a relatively small number of points that can reasonably be considered due to noise. 
Applying \texttt{DBSCAN} with ad-hoc parameters $\epsilon = 0.3$ and $\texttt{MinPts} = 50$ returns two clusters with only $0.4\%$ of the vertices labeled as noise and near-perfect $99.3\%$ accuracy on the assigned points.
After normalizing for the well-understood effects of degree fluctuations,  LiveJournal thus behaves as our theory predicts: the embeddings of each community is density-connected and \texttt{DBSCAN} correspondingly attains high accuracy.

\begin{figure}[htbp]
\centering
\includegraphics[width=0.8\textwidth]{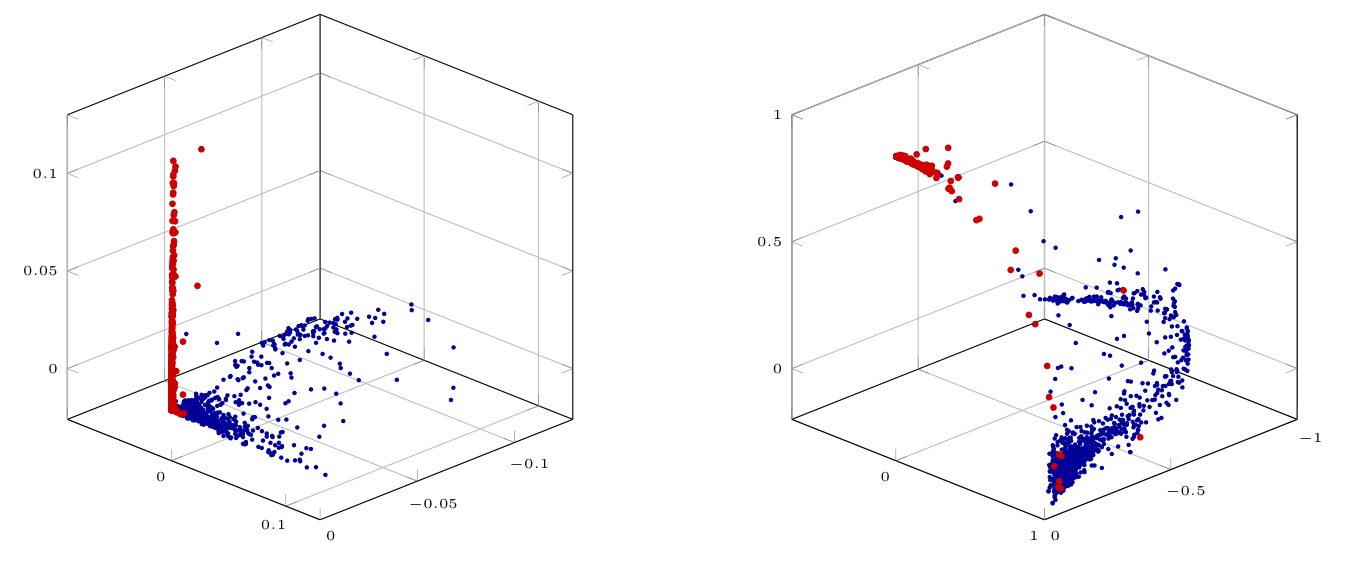}
\caption{ Spectral embedding of LiveJournal given by the rows of
$[\psi_1\ \psi_2\ \psi_3]$, the eigenvectors of the three largest eigenvalues, colored
according to the ground-truth communities.
Left: the unnormalized embedding, with axes
truncated at $\pm 0.13$ ($0.6\%$ of the points fall outside the displayed range). 
We observe radial streaks and poor separation between the communities near the origin. 
Right: the same embedding after
normalizing each row to the unit sphere. 
The two communities become well separated. }
\label{fig:livejournal_leading3_pair}
\end{figure}

\subsection{Conclusion}
This paper studied spectral clustering in the block latent space model, a random graph model with community structure in the presence of a potentially inhomogeneous confounding latent geometry.
Propositions \ref{prop: EigenvaluesMatch} and \ref{prop: EigenvectorsMatch} described the spectrum of the graph's adjacency matrix in terms of a limiting kernel. 
Theorem \ref{thm: AlgGuarantee} leveraged this structure to give a consistency guarantee for \texttt{DBSPEC}, a \texttt{DBSCAN}-based higher-order spectral clustering algorithm. 

For real-world data, our results are consistent with the observed behavior once the understood effects of large cliques and degree fluctuations are taken into account.
We found that a higher-order eigenvalue can give better performance than the classical second eigenvalue, and our theory gives good predictions for the location of the ideal eigenvalue. 
Further, robust multidimensional spectral embedding can significantly improve performance relative to sign-based assignment with a single eigenvector.

\subsection*{Acknowledgements}
KA has been supported by the French government, through the France 2030 investment plan managed by the Agence Nationale de la Recherche, as part of the Université Côte d'Azur's Initiative of Excellence, reference ANR-15-IDEX-0001 and of the ``UCA DS4H'' project, reference ANR-17-EURE-0004. LS has been supported by CAPES (Coordenação de Aperfeiçoamento de Pessoal de Nível Superior) under project MATH-AMSUD 88881.694479/2022-01 and supported by CNPq (Conselho Nacional de Desenvolvimento Científico e Tecnológico) under project 140162/2022-4. AVW is funded by the Deutsche Forschungsgemeinschaft (DFG, German Research Foundation) under Germany's Excellence Strategy EXC 2044/2--390685587, Mathematics Münster: Dynamics–Geometry–Structure, and also acknowledges support from project OCENW.KLEIN.324 of the research programme Open Competition Domain Science – M, which is partly financed by the Dutch Research Council (NWO).

\subsection*{Declarations}
The use of generative artificial-intelligence in the text portions of this manuscript was mostly restricted to occasional language editing, or help locating certain references. 
In the experimental part, Claude (Anthropic, model Fable 5) was  used to assist with writing and refactoring the analysis notebooks and generating figures. The authors reviewed all generated code.   
The notebooks were re-executed end-to-end and the authors verified the reported numbers against the computations.  
The authors assume responsibility for all content.

\newpage 
 
  \appendix
    \section{Proof for Remark \ref{rem: MutlipleClusters}}
    Recall that Remark \ref{rem: MutlipleClusters} concerns generalizations to settings with multiple clusters and non-uniform cluster label assignments. 

    \subsection{Model and performance guarantee}\label{sec: ModelDefinition2}
    As in the main text, we assume that nodes are independently assigned latent positions $X_1,\ldots,X_n$ following some probability measure $\mu$ on $\cX \subseteq \bbR^d$, and we consider some kernel $K:\cX\times \cX \to [0,1]$. 
    However, we now assume that cluster assignments $\sigma_v$ are in $\{1,\ldots,k \}$ for some $k\geq 2$, independently with $\mathbb{P}(\sigma_v = j) = q_j$ for $q_1,\ldots,q_k>0$ with $\sum_i q_i = 1$. 
    Further, fix some $\bP\in [0,1]^{k\times k}$ that is symmetric, $\bP^T = \bP$. 

    A random graph $G = (V,E)$ on $n$ nodes is then said to come from the \emph{general $k$-cluster block latent space model} if its edges are present conditionally independently given the latent positions $X_v$ and cluster assignments $\sigma_v$, with conditional probability specified by 
    \begin{align} 
        \bbP_{X,\sigma}\bigl((v,w)\in E \bigr) = 
            \rho_n K(X_v,X_w) \bP(\sigma_v,\sigma_w).\label{eq: Def_PXsigma2}
    \end{align} 

    As before, we adopt Assumptions \ref{ass: ConnectedLatent} and \ref{ass: ContinuityKernel} concerning the connectedness of $\cX$ and $L^2$-continuity of $K$.
    As a replacement for Assumption \ref{ass: Nontrivial}, we consider the following:
    \begin{assumption}\label{ass: NontrivialApx}
        The rows of $\bP$ are pairwise distinct. 
        Further, assume that for every $x\in \cX$ it holds that  
        $
        \int_{\cX} K(x,y)\, \intd\mu(y) >0 
        $.  
    \end{assumption} 
     
    Consider the matrix $\tilde{\bP}$ defined by 
    \begin{align}
        \tilde{\bP} \de  \bP \operatorname{diag}(q_1,\ldots,q_k),
    \end{align}
    where $\operatorname{diag}(\cdot)$ refers to the diagonal matrix with the specified diagonal.
    Let $\tau_1,\ldots,\tau_k\in\mathbb{R}$ be the eigenvalues of $\tilde{\bP}$. 
    The following generalizes Theorem \ref{thm: AlgGuarantee}.

\begin{theorem}\label{thm: AlgGuarantee2}
    Suppose that we are given a Borel set $B$ with 
    \begin{align}
        \{\tau_1\kappa_*,\ldots,\tau_k \kappa_*: j\leq k, \tau_j \neq 0 \}\subseteq  B.\label{eq:GoodOwl}
    \end{align}
    Assume that $B$ satisfies $\nu(\partial B) = 0$ and $0\not\in \overline{B}$. 
    Further, assume that $\rho_n = \omega(\ln(n)/n)$ and let 
    \begin{align} 
        \Lambda_n \de \{ n\rho_n y : y\in B \}. \label{eq: Def_Kn2}
    \end{align}
    Then, there exist  $c_1,c_2>0$ such that the partition $\sC$ of $\{1,\ldots,n \}$ output by Algorithm \ref{alg:high-density} with parameters $\varepsilon = c_1/\sqrt{n}$ and $\textnormal{\texttt{MinPts}} = c_2n$  satisfies the following asymptotically almost surely: 
    \begin{description}
        \item[(a)] {\bf Linearly sized clusters.} All except $k$ parts in $\sC$ have cardinality $o(n)$. The remaining $k$ parts $\cC_1,\ldots,\cC_k \in \sC$ satisfy that $\#\cC_i = q_in - o(n).$ 
        \item[(b)] {\bf Almost exact recovery.} There exists a permutation $\pi$ such that $\sigma_v =i$ for all except $o(n)$ members $v \in \cC_{\pi(i)}$ for every $i\leq k$.  
    \end{description}
\end{theorem} 
We explain the necessary modifications to the statements and proofs of Propositions \ref{prop: EigenvaluesMatch} and \ref{prop: EigenvectorsMatch} in Appendix \ref{apx: ExtendEigenvec}. 
Given those extensions, the proof of Theorem \ref{thm: AlgGuarantee2} proceeds almost identically to the proof of Theorem \ref{thm: AlgGuarantee} given in Section \ref{sec: ProofAlgGuarantee}, with a minor modification in the proof of Lemma \ref{lem: discontinuity}. 
We give the latter modification in Appendix \ref{sec: ProofThmAlgGuarantee2}.

\subsection{Extensions of Propositions \ref{prop: EigenvaluesMatch} and \ref{prop: EigenvectorsMatch}}\label{apx: ExtendEigenvec}
     Let $\sqcup_{i=1}^k \cK$ be the discrete union of $k$ copies of the latent space and define a measurable function $\tilde{\sA}:(\sqcup_{i=1}^k \cK)\times (\sqcup_{i=1}^k \cK) \to [0,1]$ by 
    \begin{align} 
        \tilde{\sA}\bigl((x_1,\sigma_1), (x_2,\sigma_2)) \de  K(x_1,x_2) \bP(\sigma_1,\sigma_2).  
    \end{align}
    Then, by definition of the connection probabilities, 
    \begin{align}
        \bbE[\bA_{v,w}  \mid X,\sigma] = 
        \begin{cases}
            \rho_n \sA\bigl((X_v, \sigma_v), (X_w, \sigma_w) \bigr) & \textnormal{ if }v\neq w, \\ 
            0 & \textnormal{ if }v=w.\label{eq:GloomyUser2}
        \end{cases}
    \end{align}
    Associated to the kernel $\tilde{\sA}$, we define an integral operator.
    For $f_1,\ldots,f_k \in L^2(\cX, \mu)$, we let  
    \begin{align}
        \tilde{\bbA}(f_1,\ldots,f_k) \de (g_1,\ldots,g_k)\ \  \textnormal{ with }\ \ g_i \de   \sum_{j=1}^k\bP(i,j)q_j \int_{\cX}  K(x,y) f_j(y) \, \intd \mu(y). 
    \end{align} 
    Note that $\tilde{\bbA} = \tilde{\bP} \otimes \bbK$. 
    Hence, $\tilde{\bbA}$ has eigenvalues of the form $\tau_j \kappa_i$.

    \begin{proposition}\label{prop: EigenvaluesMatch2}
        Adopt the notation and assumptions of Appendix \ref{sec: ModelDefinition2}. 
        Let $\hat{\lambda}_1,\ldots,\hat{\lambda}_n$ be the eigenvalues of the rescaled adjacency matrix $\bA/(n\rho_n)$ and assume that $\rho_n = \omega(\ln(n)/n)$.
        Then, there exists a measure $\tilde{\nu}$ such that for any fixed Borel set $B$ with $\nu(\partial B) = 0$ and $0\not\in \overline{B}$, we have 
        \begin{align} 
            \#\{i \leq n:  \hat{\lambda}_i \in B \}
 \to   \tilde{\nu}(B)
        \end{align}
        in probability.
        Moreover, this limiting measure $\tilde{\nu}$ is explicitly given by 
        \begin{align} 
            \tilde{\nu} = \sum_{i=1}^\infty\sum_{j=1}^k \delta_{\tau_j\kappa_i}.\label{eq:MadImp2}
        \end{align}
    \end{proposition}
    \begin{proof}
    Applying the matrix Bernstein inequality \cite[Theorem 1.6]{tropp2012user} word-for-word as in the proof of Lemma \ref{lem: MatrixBernstein}, it follows from the assumption $\rho_n = \omega(\ln(n)/n)$ that 
    \begin{align}
            \Vert  \bA  - \bbE[\bA  \mid X,\sigma]  \Vert/(n\rho_n) \to 0,\label{eq:CheekyBeetle}
    \end{align}
    in probability.

    Let $Q$ be the probability measure on $\{1,\ldots,k \}$ with weights $q_1,\ldots,q_k$. 
    Then, we observe that $\bbE[\bA_{v,w}  \mid X,\sigma]/\rho_n$ is a random matrix sampled from  the kernel $\tilde{\sA}$ on $\sqcup_{i=1}^k \cX$ when the latter space is equipped with the probability measure $ Q\otimes \mu$.
    The associated integral operator for the kernel $\tilde{\sA}$ and the measure $Q\otimes \mu$ is exactly $\tilde{\bbA}$.
        
    The remainder of the argument now follows exactly as in the proof of Proposition \ref{prop: EigenvaluesMatch}. 
    One can use \cite[Theorem 3.1]{koltchinskii2000random} to match the eigenvalues of $\bbE[\bA \mid X,\sigma]/(n\rho_n)$ with those of $\tilde{\bbA}$. 
    Further, one can use \eqref{eq:CheekyBeetle} and Weyl's perturbation inequality  \cite[Theorem 4.3.15]{horn2012matrix} to match the eigenvalues of $\bA$ with those of $\bbE[\bA \mid X,\sigma]$.
    \end{proof}
    Let $\varrho_1,\ldots,\varrho_k \in \bbR^k$ be the right eigenvectors of $\tilde{\bP}$ with eigenvalues $\tau_1,\ldots,\tau_k$, respectively. 
    Then, recalling from \eqref{eq:UneasyCamel} that we denote $\varphi_i$ for the eigenfunction of $\bbK$, the eigenfunction of the operator $\tilde{\bbA}$ are exactly of the form $\varrho_j\otimes \varphi_i$.  
    As the corresponding generalization for \eqref{eq: Def_psi} we define vectors $\tilde{\psi}_i^{(j)},\ldots,\tilde{\psi}_i^{(j)}\in \bbR^n$ by 
    \begin{align}
        (\tilde{\psi}_i^{(j)})_v \propto (\varrho_{j})_{\sigma(v)} \varphi_i(X_v),\label{eq:GiddyWisp}
    \end{align}
    with the proportionality constant chosen such that these are unit vectors.

    \begin{proposition}\label{prop: EigenvectorsMatch2}
        Adopt the notation and assumptions of Appendix \ref{sec: ModelDefinition2}. 
        Fix a Borel set $B$ with $\tilde{\nu}(B)>0$. 
        Further, assume that $\tilde{\nu}(\partial B) = 0$ and $0 \not\in \overline{B}$. 
        Let $\tilde{\psi}_B$ be the $n\times \tilde{\nu}(B)$ matrix with columns given the $\tilde{\psi}_i^{(j)}$ with $\tau_j\kappa_i \in B$, and let $\hat{\Psi}_B$ be a matrix with columns given by those eigenvectors of $\bA/(n\rho_n )$ with eigenvalues in $B$.

        Assume that $\rho_n = \omega(\ln(n)/n)$.
        Then, it holds asymptotically almost surely that $\hat{\Psi}_B$ is also $n\times \tilde{\nu}(B)$. 
        Moreover, it holds in probability that 
        \begin{align} 
            \min_{\hat{\bO} \in \mathbb{O}(\tilde{\nu}(B))}\Vert \hat{\Psi}_B \hat{\bO} - \tilde{\Psi}_B \Vert_{F} \to 0\label{eq:Procruste2} 
        \end{align}
        where the minimum runs over orthogonal matrices and $\Vert \cdot \Vert_F$ is the Frobenius norm. 
    \end{proposition}
    \begin{proof}
        This follows identically to Proposition \ref{prop: EigenvectorsMatch} by using \cite[Theorem 2.1]{koltchinskii1998asymptotics} to control the eigenvectors of $\bbE[\bA \mid  X,\sigma ]/(n\rho_n)$ in terms of those of $\tilde{\bbA}$ and then using the Davis--Kahan theorem \cite[Theorem 2]{yu2015useful} to match the eigenvectors of $\bA$ with those of $\bbE[\bA \mid  X,\sigma ]$.
    \end{proof}

    \subsection{Proof of Theorem \ref{thm: AlgGuarantee2}}\label{sec: ProofThmAlgGuarantee2}
    The main modification to the proof in Section \ref{sec: ProofAlgGuarantee} is an argument that Assumption \ref{ass: NontrivialApx} suffices to guarantee some spacing between the clusters' spectral embeddings.  
    Recall from Algorithm \ref{alg:high-density} that $\hat{V}_w = [(\hat{\psi}_i)_w :  \hat{\lambda}_i \in \Lambda_n]$ takes the $w$th entry of each eigenvector of $\bA$ with an eigenvalue in $\Lambda_n$. 
    Similarly, we define $V_w \de [(\tilde{\psi}_i^{(j)})_w : \tau_j \kappa_i \in  B]$. 
    \begin{lemma}\label{lem: discontinuity2}
        Adopt the notation and assumptions of Theorem \ref{thm: AlgGuarantee2}.
        Then, there exists some $\varepsilon_0>0$ such that such that asymptotically almost surely  $\Vert V_v -V_w \Vert > \varepsilon_0/\sqrt{n}$ whenever $\sigma_v \neq \sigma_w$. 
    \end{lemma}
    \begin{proof}
        Recall from Lemma \ref{lem: StrictPos} that the eigenfunction $\varphi_*$ associated with the greatest eigenvalue $\kappa_*$ of $\bbK$ is strictly positive. 
        In particular, there exists some $c>0$ such that $\varphi_*(x)>c$ for all $x\in \cX$. 

        The assumption that the rows of the symmetric matrix $P$ are pairwise distinct implies that for every $s,t \in \{1,\ldots,k \}$ there exists some vector $x$ in its column space with $x_{s} \neq x_{t}$. 
        Note that the column space of the matrix $\tilde{P} = P \operatorname{diag}(q_1,,\ldots,q_k)$ is equal to that of $P$ and spanned by the right eigenvectors $\rho_j$ of $\tilde{P}$ with nonzero eigenvalue $\tau_j \neq 0$.  
        It follows that there is some $\rho_j$ with nonzero eigenvalue and $(\rho_j)_{s} \neq (\rho_{j})_{t}$. 
        Let $r>0$ some small constant that lower bounds the absolute difference, say $r \de \min_{s\neq t} \max_{j:\tau_j\neq 0} \lvert (\rho_j)_{s} - (\rho_j)_{s}  \rvert$.  
        
        Recall that we assumed in \eqref{eq:GoodOwl} of Theorem \ref{thm: AlgGuarantee2} that $\tau_j\kappa_* \in B$ for all values with $\tau_j\neq 0$.
        It follows from the definition \eqref{eq:GiddyWisp} of $\psi_*^{(j)}$ that with $Z_j$ the (random) proportionality constant for rescaling to a unit vector, 
        \begin{align}
            \Vert V_v - V_w \Vert \geq \max_{j} \lvert (\psi_*^{(j)})_v - (\psi_*^{(j)})_w\rvert  =\lvert \varphi_*(X_v)  \rvert \max_{j} Z_j \lvert (\rho_{j})_{\sigma(v)} - (\rho_{j})_{\sigma(w)} \rvert \geq cr \max_{j}Z_j.    
        \end{align}
        Here, we used that $\sigma_v \neq \sigma_w$. 
        Use that $\sqrt{n}Z_j \to 1$ almost surely, by the law of large numbers, to conclude.  
    \end{proof}

    \begin{proof}[Proof of Theorem \ref{thm: AlgGuarantee2}]
    This follows identically to the proof of Theorem \ref{thm: AlgGuarantee} given in Section \ref{sec: ProofAlgGuarantee}, except that we use Lemma \ref{lem: discontinuity2} in place of Lemma \ref{lem: discontinuity}.        
    \end{proof}

    \section{Additional synthetic experiments}

    We continue discussing the model from Example~\ref{example: Geometric}. 
    This example serves to highlight that $\mathcal{X}$ does not need to be a torus, the latter being a common assumption in previous works.
    
    Consider two communities with 1500 nodes each. 
    We want to apply our algorithm to recover communities of the block latent-space model with the following parameters $R=0.35$, $a=0.75$ and $b=0.2$ and to investigate the effect of the underlying topology. 
    In the scenario with the square space $\mathcal{X}=[0,1]^2$, the ideal eigenvalue is $\hat{\lambda}_*=256.52$.
    The eigenvalue $\hat{\lambda}_4=255.08$ is closest to $\hat{\lambda}_*$ and we depict the coordinates of $\psi_4$ in Figure~\ref{fig:02} (left). From Figure~\ref{fig:02} (left) it is easy to see that Algorithm~\ref{alg:high-density} can accurately classify all nodes of the graph. Figure~\ref{fig:02} (right)  depicts the coordinates of the eigenvector associated with $\lambda_2=308.87$, which clearly would not recover the true community labels, showing that the standard spectral clustering does not work in this example. 

\begin{figure}[htbp]
    \centering
    \includegraphics[width=\textwidth]{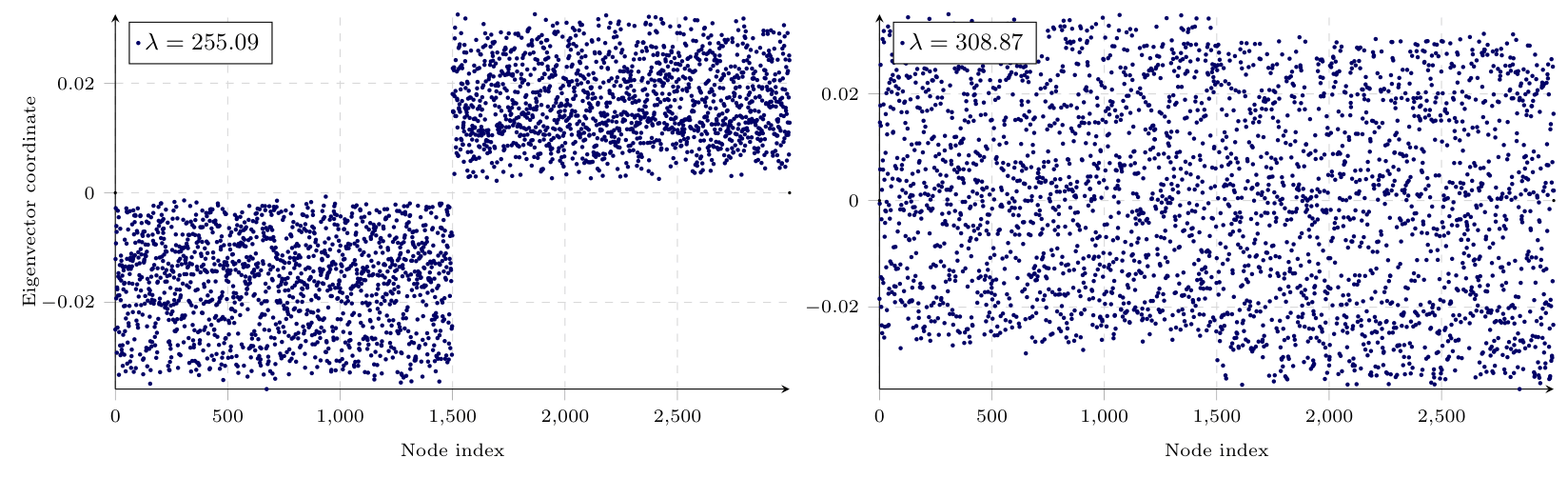}
    \caption{On the left, plot of the eigenvector coordinates associated with $\hat{\lambda}_4 = 255.09$. On the right, the eigenvector coordinates associated with $\hat{\lambda}_2 = 308.87$.}
    \label{fig:02}
\end{figure}

To see how the ideal eigenvalue changes with the geometry, we now consider the same parameters, but with $\cX = \mathbb{T}^2$, the unitary two-dimensional torus. We now have $\hat{\lambda}_*=317.26$. Also, Figure~\ref{fig:02atorus} depicts the coordinates of the eigenvectors associated with  $\hat{\lambda}_2=318.81$ (left) and  $\hat{\lambda}_4=281.19$ (right), respectively. One can see that $\psi_2$ recover the true communities. 
Compared with the previous example, this shows that the location of the optimal eigenvalue depends on the underlying geometry and may or may not coincide with the second eigenvalue.

\begin{figure}[htbp]
    \centering
    \includegraphics[width=\textwidth]{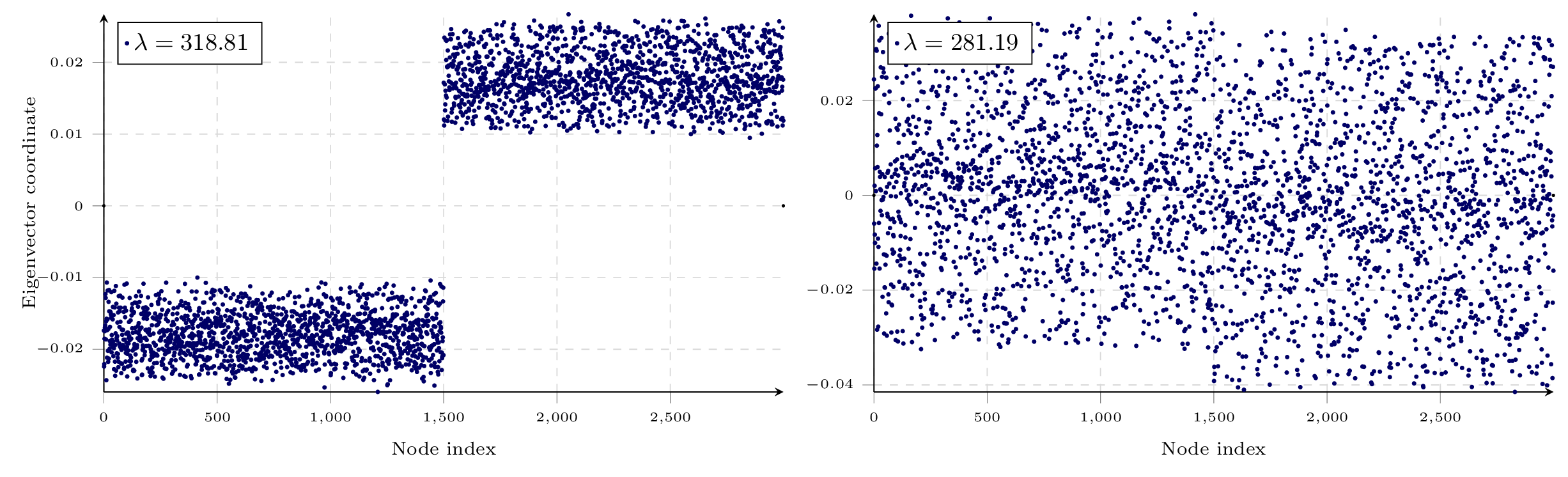}
    \caption{On the left, plot of the eigenvector coordinates of $\psi_2$ associated with $\lambda_2=318.81$. On the right, plot of the eigenvector coordinates of $\psi_4$ associated with $\lambda_4=281.19$.}
    \label{fig:02atorus}
\end{figure}
   
  We now examine a regime where $\lambda_*$ is a multiple eigenvalue in the limiting spectrum. As shown in equation (10) of \cite{avrachenkov2021higher}, this situation occurs for specific values of the model parameters. 
  (See also Appendix \ref{sec:OnParticularModel}.)
  In this setting, we compare the algorithms \texttt{HOSC}~\cite{avrachenkov2021higher} and \texttt{DBSPEC}. 
  Let us continue with $\cX = \mathbb{T}^2$ and set $n=3000$, $R=1/2\pi$, $\rho_n=1$,   and $K(x,y)=\bb1\{||x-y||\leq R\}$. We vary $a$ from $0.15$ to $1$, with step size $0.05$. For each $a$, we let $b$ be given by
   \begin{align}
    b = a\frac{1-\operatorname{sinc}(2\pi R)}{1+\operatorname{sinc}(2\pi R)}, \label{eq:AngryLeaf}
   \end{align}
    the ideal eigenvalue $\lambda_*$ is non-simple in the limiting spectrum; see equation (\ref{eq:plus_cond2}) in Appendix~\ref{sec:OnParticularModel}.
    
    Figure~\ref{fig:05} depicts the evolution of accuracy for \texttt{HOSC} (blue) and \texttt{DBSPEC} (red) for such pairs $(a,b)$. Here \texttt{DBSPEC} is run with the three eigenvectors whose eigenvalues are closest to $\lambda_*$, and its density-based step is \texttt{DBSCAN} with $\epsilon=0.01$ and $\mathrm{MinPts}=10$ in the Euclidean metric. 
    Reported accuracies are averaged over 30 independent runs, and bars indicate standard deviation for the individual samples. 
    We note that \texttt{DBSPEC} recovers the true communities for every pair $(a,b)$ consistently, while \texttt{HOSC} accuracy fluctuates wildly.  
    This confirms the benefit of the new algorithm being robust to non-simple eigenvalues.

    \begin{figure}[htbp]
    \centering
    \includegraphics[width=0.6\textwidth]{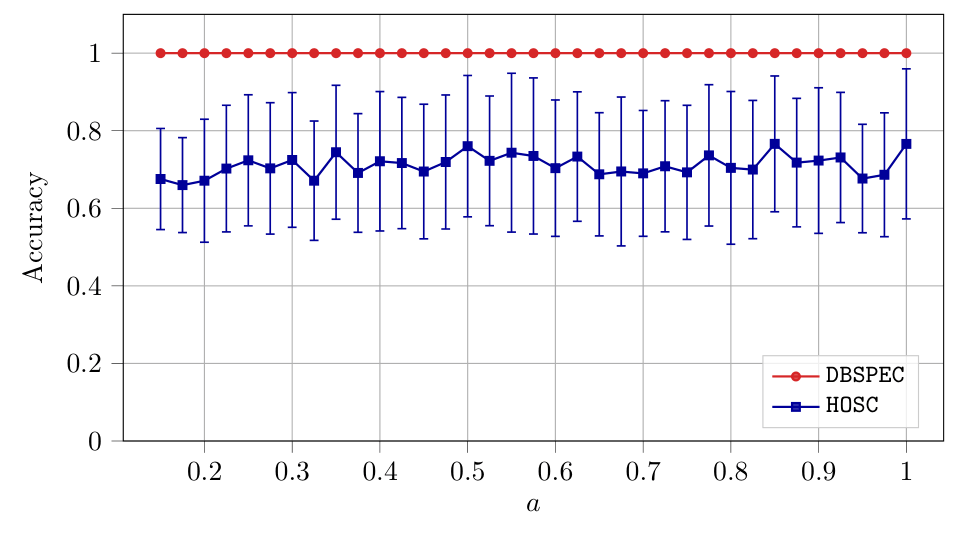}
      \caption{Accuracy of SGBM for  \texttt{HOSC}~\cite{avrachenkov2021higher} and \texttt{DBSPEC} with $b$ is chosen depending on $a$ as in \eqref{eq:AngryLeaf}. Due to the linearity of \eqref{eq:AngryLeaf}, changing $a$ is here equivalent to modifying the sparsity parameter $\rho_n$. The bars indicate one standard deviation for the individual samples across runs.}   
      \label{fig:05}
    \end{figure}

    Finally let us consider an example of sparse graphs.
    Specifically, we take $\rho_n=10\cdot\frac{(\log n)^2}{n}$ and vary  $n$ from $500$ to $10000$. The other parameters of SGBM are given by $a=1/2$ and $b=1/5$. Figure~\ref{fig:06} depicts the evolution of the accuracy and the density as $n$ grows. Reported accuracies are averaged over 50 independent runs. We observe that \texttt{DBSPEC} recovers the true community clusters even for modest size of $n$.

    \begin{figure}[htbp]
    \centering
    \includegraphics[width=0.8\textwidth]{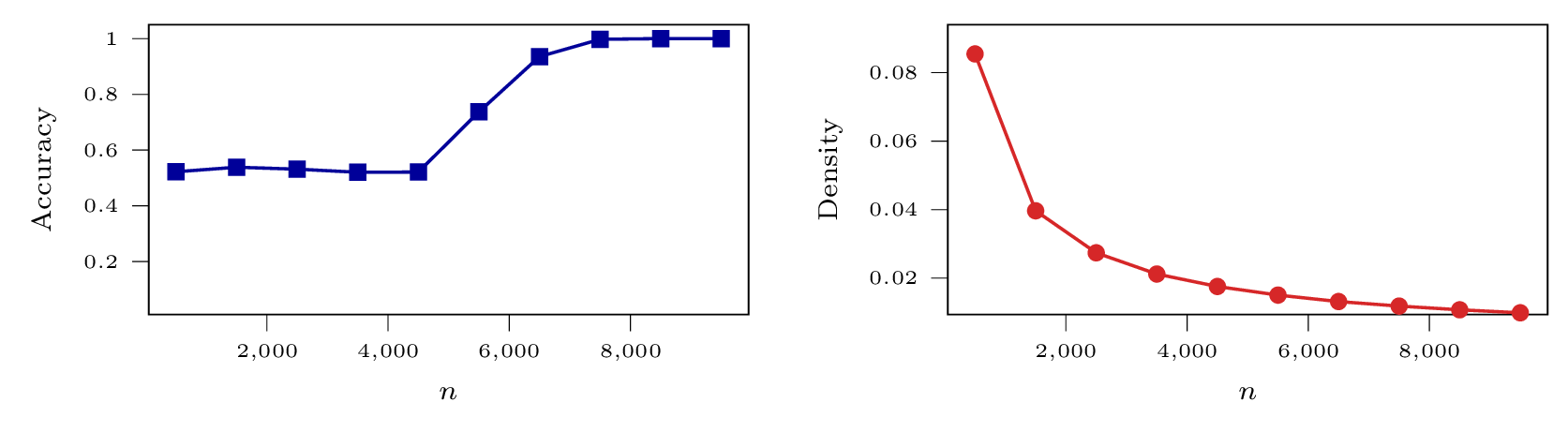}
      \caption{On the left, evolution of \texttt{DBSPEC} accuracy in the sparse regime as $n$ increases, that is, the proportion of correctly classified vertices. On the right, the average density for each $n$, where the density is the ratio between the number of edges and the total number of possible edges.}

      \label{fig:06}
    \end{figure}

\section{When the ideal eigenvalue is non-simple in the limiting spectrum}\label{sec:OnParticularModel}

Let us consider a model similar to Example~\ref{example: Geometric}, with the only difference in the choice $\mathcal{X}=\mathbb{T}^2$, the unitary torus, which can be represented by $\left[-\frac{1}{2},\frac{1}{2}\right]^2$ with the toroidal metric and specify a condition of the ideal eigenvalue multiplicity. 
As in \cite{avrachenkov2021higher}, we define
$$F_{\text{in}}(x) = a \cdot \bb1\{\Vert x-y \Vert \leq R\}
\quad \mbox{and} \quad
F_{\text{out}}(x) = b \cdot \bb1\{\Vert x-y \Vert \leq R\},
$$
for some $b<a \in (0,1]$,
$$
\mu_{in}=\int_{\mathbb{T}^2 }F_{in}(x)dx=a(2R)^2
\quad \mbox{and} \quad 
\mu_{out}=\int_{\mathbb{T}^2 }F_{out}(x)dx=b(2R)^2,$$
and
$$
\Hat{F}_{\mathrm{in}}(k)=\int_{\mathbb{T}^2}F_{\mathrm{in}}(x)e^{-2i\pi \langle k,x\rangle}dx
\quad \mbox{and} \quad
\Hat{F}_{\mathrm{out}}(k)=\int_{\mathbb{T}^2}F_{\mathrm{out}}(x)e^{-2i\pi \langle k,x\rangle}dx,
$$
the Fourier transforms of $F_{\mathrm{in}}(x)$ and $F_{\mathrm{out}}(x)$.

From Theorem~1 and Proposition~1 of \cite{avrachenkov2021higher}, we know that the ideal eigenvalue $\lambda_*$ will not be simple in the limiting spectrum, if and only if
\begin{equation}
\label{eq:plus_cond}
\Hat{F}_{\mathrm{in}}(k)+\Hat{F}_{\mathrm{out}}(k) = \mu_{\mathrm{in}}-\mu_{\mathrm{out}},
\end{equation}
for some $k \in \mathbb{Z}^2 \backslash \{0\}$. Using Lemma~3 from~\cite{avrachenkov2021higher} for the expression of the Fourier transform, we can specify the above condition to
\begin{equation}  
\label{eq:plus_cond1}
(2R)^2 (a+b) \prod_{j=1}^2 \operatorname{sinc}(2\pi k_j R) = (2R)^d (a-b),
\end{equation}
for some $k \in \mathbb{Z}^2 \backslash \{0\}$, where 
$$
\operatorname{sinc}(x) = \begin{cases}
        \frac{\sin x}{x}, \quad  \text{if } x\neq 0; \\
 1, \quad \text{otherwise.}\\
    \end{cases}
$$
In particular, by taking $k=(0,1)$, the ideal eigenvalue will be non-simple in the limiting spectrum if 
\begin{equation}  
\label{eq:plus_cond2}
b = a \frac{1-\operatorname{sinc}(2\pi R)}{1+\operatorname{sinc}(2\pi R)}.
\end{equation}

\section{Additional experiments with \texttt{DBSPEC} on real-world data}
We here consider the performance of \texttt{DBSPEC} with multi-dimensional spectral embeddings for DBLP A$^\ast$. 

We consider the three eigenvectors $\psi_8$, $\psi_9$, and $\psi_{10}$, whose associated eigenvalues are closest to our estimator
$\hat{\lambda}_*=14.31$, namely $\hat{\lambda}_8=14.92$, $\hat{\lambda}_9=14.84$ and
$\hat{\lambda}_{10}=14.63$, respectively. 
The next eigenvalue, $\hat{\lambda}_{11}=13.87$ is separated from $\hat{\lambda}_{10}$ by a comparatively wide
gap.

Figure~\ref{fig:dblp_4clusters} shows the spectral embedding given by the rows of $[\psi_8\ \psi_9\ \psi_{10}]$, after normalizing each embedded point to lie on the unit sphere. Interestingly, for a wide range of \texttt{DBSCAN} parameters (for instance $\epsilon = 0.30$ and $\texttt{MinPts} = 500$) this embedding yields four clusters rather than two. Among the 12,991 embedded points, the largest cluster contains 7,702 points, of which $84\%$ belong to the first community, while the other three clusters, with 1,657, 1,312 and 922 points, consist almost entirely of vertices of the second
community ($89\%$, $93\%$ and $83\%$ of their points, respectively), the remaining 1,398 points ($10.8\%$) are labeled as noise. In particular, if the latter three clusters are merged into a single one, the resulting
bipartition attains accuracy $0.86$ when the noise points are removed. The adjacency matrix may also suggest that four clusters are more natural for this dataset, since reordering the rows and columns of the adjacency matrix by the four clusters reveals four dense diagonal blocks, with only $3.6\%$ of the edges crossing a block boundary (Figure~\ref{fig:dblp_matrix}).

\begin{figure}[htbp]
\centering
\includegraphics[width=1\textwidth]{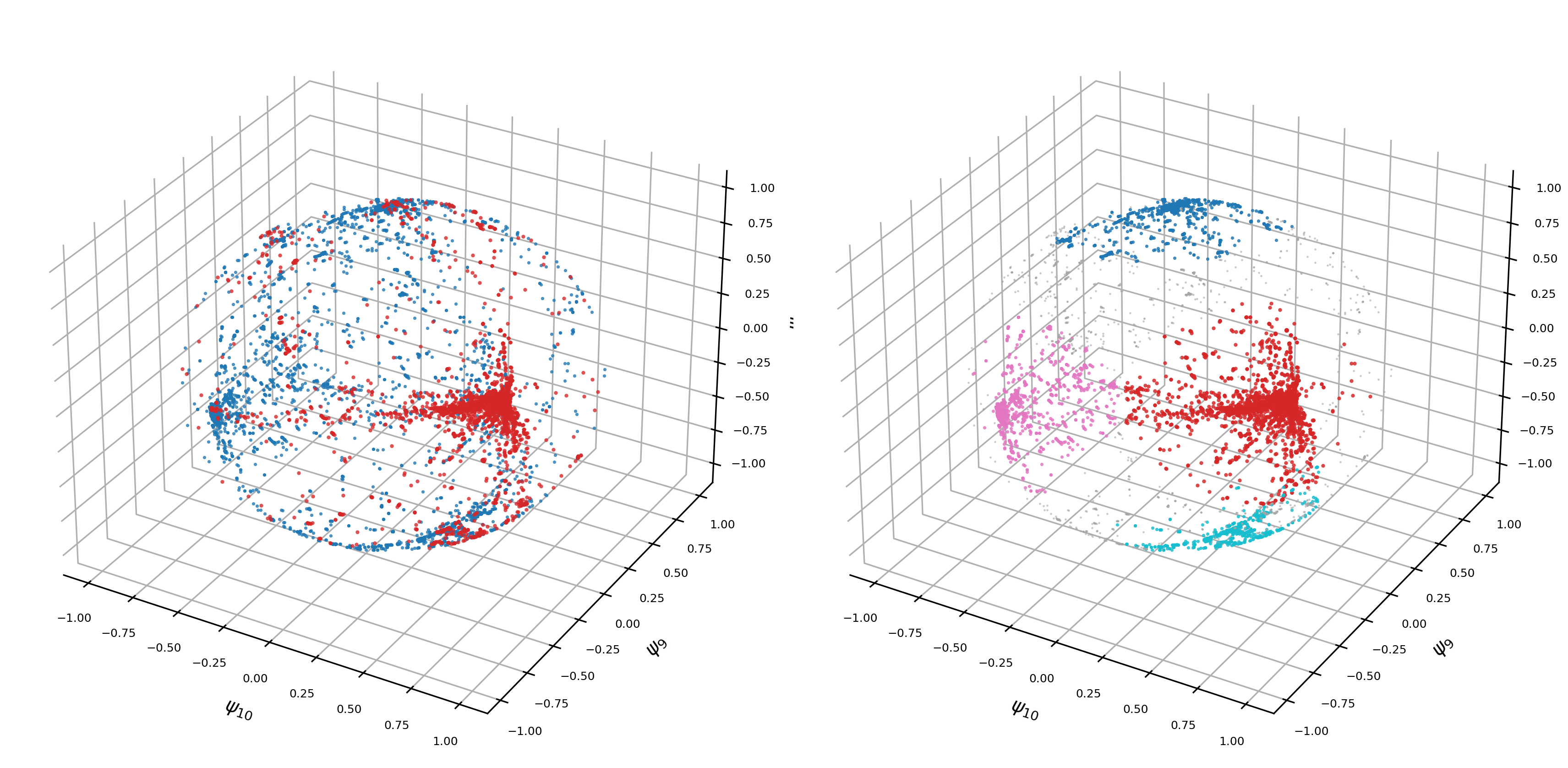}
    \caption{ Three dimensional spectral embedding based on the eigenvectors $\psi_8$, $\psi_9$, and $\psi_{10}$, with each embedded point normalized to lie on the unit sphere. On the left, points are colored according to the true communities. On the right, points are colored according to the four clusters obtained by \texttt{DBSCAN} with $\epsilon=0.30$ and $\texttt{MinPts}=500$, with noise points shown in gray ($10.8\%$ of the vertices).}
    \label{fig:dblp_4clusters}
\end{figure}

\begin{figure}[h]
\centering
\includegraphics[width=0.3\textwidth]{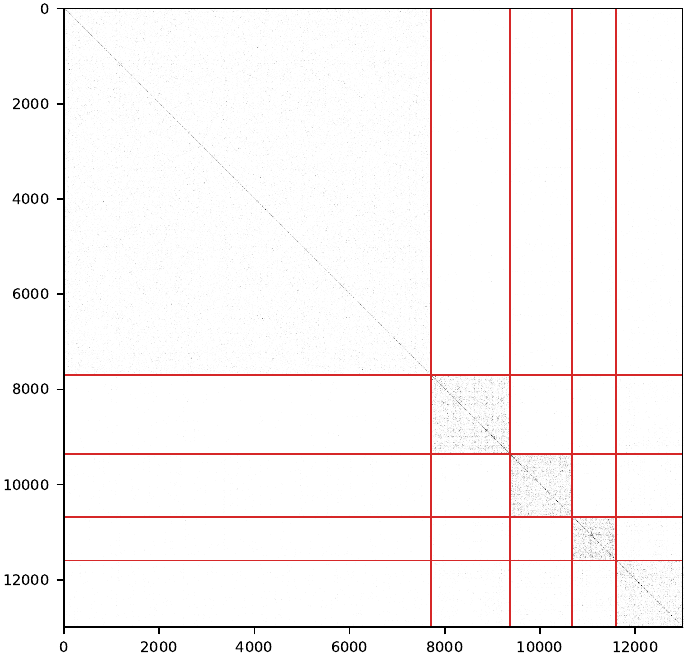}
\caption{Adjacency matrix of DBLP A$^\ast$ with rows and columns sorted by the four
\texttt{DBSCAN} clusters of Figure~\ref{fig:dblp_4clusters} (noise vertices last). Red lines
mark the cluster boundaries. The three subgroups of the second community appear as
dense diagonal blocks; only $3.6\%$ of the edges among the assigned vertices cross a
block boundary.}
\label{fig:dblp_matrix}
\end{figure}

This finding of four natural clusters is somewhat surprising, since the dataset was explicitly constructed with two ground-truth clusters. 
In DBLP, however, these ground-truth clusters are based on publication venues. 
Moreover, the variant DBLP A$^{\ast}$ is based on the two largest such venues.
One plausible hypothesis is that the four clusters correspond to natural subcommunities in these large venues.

Let us finally recall that the main text considered spectral embeddings for LiveJournal based on the rows of $[\psi_1\ \psi_2\ \psi_3]$, the eigenvectors of the three largest eigenvalues. 
The choice of these three eigenvalues was motivates as being natural when one does not have access to information about the ideal eigenvalues, but it is natural to wonder what would happen if a different choice of eigenvectors were used. 

Figure \ref{fig:livejournal} visualizes the (normalized) spectral embeddings associated to the rows of $[\psi_3\ \psi_4\ \psi_5]$ in LiveJournal. 
The accuracy of \texttt{DBSCAN} in this case is $99.2\%$, essentially identical to the performance $99.3\%$ achieved using $[\psi_1\ \psi_2\ \psi_3]$.
This confirms that the accuracy is not too sensitive to the choice of the set of eigenvectors, as long as the set includes the ideal eigenvectors.

\begin{figure}[h]
\centering
\includegraphics[width=0.45\textwidth]{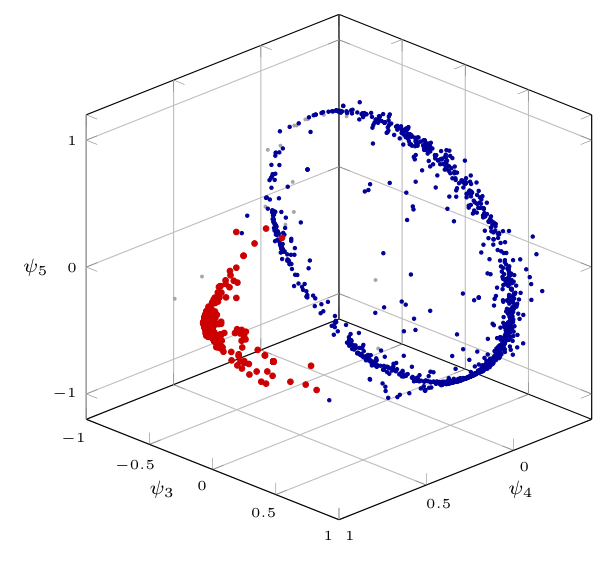}
\caption{Spectral embedding of LiveJournal given by the rows of $[\psi_3\ \psi_4\ \psi_5]$,
normalized to the unit sphere, colored by the labels returned by \texttt{DBSCAN} with
$\epsilon=0.35$ and $\text{MinPts}=50$ (gray = noise, $0.6\%$ of the vertices). }
\label{fig:livejournal}
\end{figure}

\end{document}